\documentclass[11pt]{article}

\usepackage{amssymb,fullpage,tikz,amsmath,amsthm,xspace}

\DeclareMathOperator*{\E}{\mathbb{E}}

\newcommand{\opt}{\textsc{Opt}}

\newcommand{\Opt}{\ensuremath{{\textsc{Opt}}}\xspace}

 \newtheorem{theorem}{Theorem}
 \newtheorem{observation}[theorem]{Observation}
 \newtheorem{lemma}[theorem]{Lemma}
 \newtheorem{corollary}[theorem]{Corollary}
 \newtheorem{proposition}[theorem]{Proposition}

\begin{document}

 \title{Online Computation with Untrusted Advice\thanks{A preliminary version of this paper appeared in the Proceedings of
  	the 11th Innovations in Theoretical Computer Science (ITCS), 2020. Part of the work was done while 2nd author was affiliated with CMM. Partially supported by the CNRS PEPS project ADVICE, the project PREDICTIONS, grant ANR-23-CE48-0010, the projet Algoridam, grant ANR-19-CE48-0016 from the French National Research Agency, the Center for Mathematical Modeling grant ANID FB210005, and the grant DGECR-2018-00059 from the Natural Sciences and Engineering Research Council of Canada (NSERC).}}

\author{%
Spyros Angelopoulos\thanks{Sorbonne Université, CNRS, LIP6, Paris, France}
\and
Christoph D\"urr\footnotemark[2]
\and
Shendan Jin\footnotemark[2]
\and
Shahin Kamali\thanks{Department of Electrical Engineering and Computer Science, York University, Toronto, Canada}
\and
Marc Renault\thanks{University of Wisconsin-Madison, Madison, USA}}

\date{}

\maketitle

\begin{abstract}

We study a generalization of the advice complexity model of online computation in which the advice is provided by an untrusted source. Our objective is to quantify the impact of untrusted advice so as to design and analyze online algorithms that are robust if the advice is adversarial, and efficient is the advice is foolproof. We focus on four well-studied online problems, namely ski rental, online bidding, bin packing and list update. For ski rental and online bidding, we show how to obtain algorithms that are Pareto-optimal with respect to the competitive ratios achieved, whereas for bin packing and list update, we give online algorithms with worst-case tradeoffs in their competitiveness, depending on whether the advice is trusted or adversarial. More importantly, we demonstrate how to prove lower bounds, within this model, on the tradeoff between the number of advice bits and the competitiveness of any online algorithm.  
\end{abstract}


\paragraph*{keywords}	
	Online computation ; competitive analysis ; advice complexity ; robust algorithms ; untrusted advice 


\section{Introduction}
\label{sec:introduction}

Suppose that you have an investment account with a significant amount in it, and that your financial institution advises you periodically on investments. One day, your banker informs you that company X will soon receive a big boost, and advises to use the entire account to buy stocks. If you were to completely trust the banker’s advice, there are naturally two possibilities: either the advice will prove correct (which would be great) or it will prove wrong (which would be catastrophic). A prudent customer would take this advice with a grain of salt, and would not be willing to risk everything. In general, our understanding of advice is that it entails {\em knowledge that is not foolproof}.

In this work we focus on the online computation with advice. Our motivation stems from observing that, unlike the real world, the advice under the known models is often closer to ``fiat'' than ``recommendation''. Our objective is to propose a model which allows the possibility of incorrect advice, with the objective of obtaining more realistic and robust online  algorithms.

\subsection{Online computation and advice complexity}
In the standard model of online computation that goes back to the seminal work of Sleator and Tarjan~\cite{SleTar85}, an online algorithm receives as
input a sequence of {\em requests}. For each request in this sequence, the algorithm must make an irrevocable decision concerning the item, without any knowledge of future requests. The performance of an online algorithm is usually evaluated by means of the competitive ratio, which is the
worst-case ratio of the cost incurred by the algorithm (assuming a minimization problem) to the cost of an ideal solution that knows the entire sequence in advance.

In practice, however, online algorithms are often provided with some (limited)
knowledge of the input, such as lookahead on some of the upcoming requests, or
knowledge of the input size. While competitive analysis is still applicable,
especially from the point of view of the analysis of a known, given algorithm,
a new model was required to formally quantify the power and limitations of
offline information. The term {\em advice complexity} was first coined by
Dobrev {\em et al.}~\cite{DobKraPar09}, and subsequent formal models were
presented by B\"ockenhauer {\em et al.}~\cite{BocKomKra09} and Emek {\em et
al.}~\cite{EmekFraKorRos2011}, with this goal in mind.  More precisely, in the
advice setting, the online algorithm receives some bits that encode
information concerning the sequence of input items. As expected, this
additional information can boost the performance of the algorithm, which is
often reflected in better competitive ratios.

Under the current models, the advice bits can encode any information about the input sequence; indeed, defining the ``right'' information to be conveyed to the algorithm plays an important role in obtaining better online algorithms. Clearly, the performance of the online algorithm can only improve with larger number of advice bits. The objective is thus to identify the exact trade-offs between the size of the advice and the performance of the algorithm. This is meant to provide a smooth transition between the purely online world (nothing is known about the input) and the purely ``offline'' world (everything is known about the input).
In the last decade, a substantial number of online optimization problems have been studied in the advice model; we refer the reader to the survey of
Boyar {\em et al.}~\cite{Boyar:survey:2016} for an in-depth discussion of developments in this field.

As argued in detail in~\cite{Boyar:survey:2016}, there are compelling reasons to study the advice complexity of online computation.
Lower bounds establish strict limitations on the power of any online algorithm; there are strong connections between randomized online algorithms and online algorithms with advice (see, e.g.,~\cite{DBLP:conf/icalp/Mikkelsen16}); online algorithms with advice can be of practical interest in settings in which it is feasible to run multiple algorithms and output the best solution (see~\cite{KamLopDCC14} about obtaining improved data compression algorithms by means of list update algorithms with advice); and the first complexity classes for online computation have been based on advice complexity~\cite{BoyarFavrholdt:15:Advice}.

Notwithstanding such interesting attributes, the known advice model has certain drawbacks. The advice is always assumed to be some error-free information that may be used to encode some property often explicitly connected to the optimal solution. In many settings, one can argue that such information cannot be readily available, which implies that the resulting algorithms are often impractical.

\subsection{Online computation with untrusted advice}

In this work, we address what is a significant drawback in the online advice model. Namely, all previous works assume that advice is, in all circumstances, completely trustworthy, and precisely as defined by the algorithm. Since the advice is infallible, no reasonable online algorithm with advice would choose to ignore the advice.

It should be fairly clear that such assumptions are very unrealistic or undesirable. Advice bits, as all information, are prone to transmission errors. In addition, the known advice models often allow
information that one may arguably consider unrealistic, e.g., an encoding of some part of the offline optimal solution. Last, and perhaps more significantly, a malicious entity that takes control of the advice oracle can have a catastrophic impact. For a very simple example, consider the well-known ski rental problem: this is a simple, yet fundamental resource allocation, in which we have to decide ahead of time whether to rent or buy equipment without knowing the time horizon in advance. In the traditional advice model, one bit suffices to be optimal: 0 for renting throughout the horizon, 1 for buying right away. However, if this bit is wrong, then the online algorithm has unbounded competitive ratio, i.e., can perform extremely badly. In contrast, an online algorithm that does not use advice at all has competitive ratio at most 2, i.e., its output can be at most twice as costly as the optimal one.

The above observations were recently made in the context of online algorithms with machine-learned predictions.
Lykouris and Vassilvitskii~\cite{DBLP:conf/icml/LykourisV18} and Purohit {\em et al.}~\cite{NIPS2018_8174} show how to use predictors to design and analyze algorithms with two properties: (i) if the predictor is good, then the online algorithm should perform close to the best offline algorithm (what is called {\em consistency}); and (ii) if the predictor is bad, then the online algorithm should gracefully degrade, i.e., its performance should be close to that of the online algorithm without predictions (what is called {\em robustness}). Since these works, online algorithms with predictions have been studied in a variety of settings, see, e.g., the survey~\cite{DBLP:books/cu/20/MitzenmacherV20}.

Motivated by these definitions from machine learning, in this work we analyze online algorithms based on their performance in both settings of trusted and untrusted advice. In particular, we will characterize the performance of an online algorithm $A$ by a pair of competitive ratios, denoted by $(r_A, w_A)$, respectively. Here, $r_A$ is the competitive ratio achieved assuming that the advice encodes precisely what it is meant to capture; we call this ratio the competitive ratio with {\em trusted} (thus, always correct) advice, or trusted competitive ratio of $A$, for simplicity. In contrast, $w_A$ is the competitive ratio of $A$ when the advice is untrusted, and thus potentially wrong; we will call this measure the untrusted competitive ratio of $A$. More precisely, in accordance with the worst-case nature of competitive analysis, we allow the incorrect advice to be chosen {\em adversarially}. Namely, assuming a deterministic online algorithm $A$, the incorrect advice string is generated by a malicious, adversarial entity.

To formalize the above concept, assume the standard advice model, in which a deterministic online algorithm $A$ processes a sequence of requests
$\sigma=(\sigma[i])_{i\in [1,n]}$ using an advice tape. The advice tape stores an advice string, which we denote by $\phi\in\{0,1\}^*$.  At each time $t$, $A$ serves request $\sigma[t]$,
and its output is a function of $\sigma[1, \ldots ,t-1]$ and $\phi$. Let $A(\sigma, \phi)$ denote the cost incurred by $A$ on input $\sigma$, using advice $\phi$.  Denote by $r_A$, $w_A$ as
\begin{equation}
r_A= \sup_{\sigma} \inf_{\phi} \frac{A(\sigma,\phi)}{\opt(\sigma)} ,
      \quad  \textrm{and} \quad
w_A = \sup_{\sigma} \sup_{\phi} \frac{A(\sigma,\phi)}{\opt(\sigma)},
\label{eq:definition}
\end{equation}
where $\opt(\sigma)$ denotes the optimal offline cost for $\sigma$.
Then we say that algorithm $A$ is $(r,w)$-competitive for every $r\geq r_A$ and $w\geq w_A$.
In addition, we say that $A$ has advice complexity $s(n)$ if for every request sequence $\sigma$ of length $n$, the algorithm $A$ depends only on the first $s(n)$ bits of the advice string $\phi$.
To illustrate this definition, the opportunistic 1-bit advice algorithm for ski rental that was described above is $(1,\infty)$-competitive, whereas the standard competitively optimal algorithm without advice is $(2,2)$-competitive. In general, every online algorithm $A$ without advice or ignoring its advice is trivially $(w,w)$-competitive, where $w$ is the competitive ratio of $A$. 

Hence, we can associate every algorithm $A$ to a point in the 2-dimensional space with coordinates $(r_A,w_A)$. These points are in general incomparable,
e.g., it is difficult to argue that a $(2,10)$-competitive algorithm is better than a $(4,8)$-competitive algorithm. However, one can
appeal to the notion of {\em dominance}, by saying that algorithm $A$ dominates algorithm $B$ if $r_A\leq r_B$ and $w_A\leq w_B$.
More precisely, we are interested in finding the Pareto frontier in this representation of all online algorithms. For the ski rental example, the two above mentioned algorithms belong to the Pareto set.

A natural goal is to describe this Pareto frontier, which in general may be comprised of several algorithms with vastly different statements.
Ideally, however, one would like to characterize it by a single {\em family} $\cal A$ of algorithms, with similar statements (e.g., algorithms in $\cal A$ are obtained by appropriately selecting a parameter). We say that $\cal A$ is {\em Pareto-optimal} if it consists of pairwise incomparable algorithms,
and for every algorithm $B$, there exists $A \in \cal A$ such that $A$ dominates $B$.  Regardless of optimality, given $\cal A$, we will describe its competitiveness by means of a function $f$
such that for every $r$ there is an $(r,f(r))$-competitive algorithm in $\cal A$. At some places,  we will equivalently describe its competitiveness by a function $g:\mathbb{R}_{\geq1} \rightarrow \mathbb{R}_{\geq1}$ such that for every $w$, there is an $(g(w),w)$-competitive algorithm in $\cal A$. These functions will in general depend on parameters of the problem, such as, for example, the buying cost $B$ in the ski rental problem. Note also that $r$ and $w$ may be subject to constraints; e.g., $w$ cannot be smaller than the best competitive ratio without advice. 

\subsection{Contribution}
We study various online problems in the setting of untrusted advice. We also demonstrate that it is possible to establish
both upper and lower bounds on the tradeoff between the size of the advice and the competitiveness in this new advice model.
We begin in Section~\ref{sec:ski.rental} with a simple, yet illustrative online problem as a case study, namely the {\em ski rental} problem.
Here, we give a Pareto-optimal algorithm with only one bit of advice. We also show that this algorithm is Pareto-optimal even in the space of all (deterministic) algorithms with advice of {\em any} size.

In Section~\ref{sec:online.bidding} we study the {\em online bidding} problem, in which the objective is to guess
an unknown, hidden value, using a sequence of bids. This problem was introduced in~\cite{ChrKen06} as a vehicle for formalizing efficient
doubling, and has applications in several important online and offline optimization problems.
As with ski rental, this is another problem for which a trivial online algorithm is $(1,\infty)$-competitive.
We first show how to find a Pareto-optimal strategy, when the advice encodes the hidden value, and thus can have unbounded size.
Moreover, we study the competitiveness of the problem with only $k$ bits of advice, for some fixed $k$, and
show both upper and lower bounds on the achieved competitive ratios. The results illustrate that it is
possible to obtain non-trivial lower bounds on the competitive ratios, in terms of the advice size. In particular, the lower bound implies that,
unlike the ski rental problem, Pareto-optimality is not possible with a bounded number of advice bits.

In Sections~\ref{sect:bp} and~\ref{sect:lu}, we study the \emph {bin packing} and \emph {list update} problems; these problems are central in the analysis of online problems and competitiveness, and have  
numerous applications in practice. For these problems, an efficient advice scheme should address the issues of ``what constitutes good advice'' as well as
``how the advice should be used by the algorithm''. We observe that the existing algorithms with advice perform poorly in the case the advice is untrusted. To address this, we give algorithms that can be ``tuned'' based on how much we are willing to trust the advice. This enables us to show
guarantees in the form $(r,f(r))$-competitiveness, where $r$ is strictly better than the competitive ratio of all deterministic online algorithms and $f(r)$ smoothly decreases as $r$ grows, while still being close to the worst-case competitive ratio. To illustrate this, consider the bin packing problem. Our $(r,f(r))$-competitive algorithm has 
$f(r) = \max\{33-18r,7/4\}$ 
for any $r \in [1.5,1.73]$.
If $r=1.5$, our algorithm is $(1.5,6)$-competitive, 
and matches the performance of a known algorithm~\cite{BoyarKLL16}. 
However, with a slight increase of $r$, one can 
improve competitiveness in the event the advice is untrusted. For instance,
choosing $r=1.55$, we obtain 
$f(r)=5.1$.
 In other words, the algorithm designer can hedge against untrusted advice, by a small sacrifice in the trusted performance. Thus we can interpret $r$ as the ``risk'' for trusting the advice: the smaller the $r$, the bigger the risk.
Likewise, for the list update problem, our $(r,f(r))$-competitive algorithm has $f(r)= 2 + \frac{10-3r}{9r-5}$ for $r \in [5/3,2]$. If the algorithm takes maximum risk,
i.e., if $r$ is smallest, the algorithm is equivalent to an existing $(5/3,5/2)$-competitive algorithm~\cite{BoyKamLata14}. Again, by increasing $r$, we better safeguard against the event of untrusted advice.

All the above results pertain to deterministic online algorithms. In Section~\ref{sec:extensions}, we study the power of randomization in online computation with untrusted advice. First, we show that the randomized algorithm of Purohit et al.~\cite{NIPS2018_8174} for the ski rental problem Pareto-dominates any deterministic algorithm, even when the latter is allowed unbounded advice. 
Furthermore, we show an interesting difference between the standard advice model and the model we introduce: in the former, an advice bit can be at least as powerful as a random bit, since an advice bit can effectively simulate any efficient choice of a random bit.  In contrast, we show that in our model,  there are situations in which a randomized algorithm with $L$ advice bits and one random bit is Pareto-incomparable to the Pareto-optimal deterministic algorithm with $L+1$ advice bits.
This confirms the intuition that a random bit is considered trusted, and thus not obviously inferior to an advice bit.

While our work addresses issues similar to~\cite{DBLP:conf/icml/LykourisV18} and~\cite{NIPS2018_8174}, in that trusted advice is related to consistency whereas untrusted advice is related to robustness, it differs in two significant aspects: First, our ideal objective is to identify an optimal
family of algorithms, and we show that in some cases (ski rental, online bidding), this is indeed possible; when this is not easy or possible, we can still provide approximations. Note that finding a Pareto-optimal family of algorithms presupposes that the exact competitiveness of the online problem with no advice is known. For problems such as bin packing,
the exact optimal competitive ratios are not known. Hence, a certain degree of approximation is unavoidable in such cases. In contrast,~\cite{DBLP:conf/icml/LykourisV18,NIPS2018_8174} focus on ``smooth'' tradeoffs between the trusted and untrusted competitive ratios, but do not address the issues related to optimality and approximability of these tradeoffs.


Second, our model considers the size of advice and its impact on the algorithm's performance, which is the main focus of the advice complexity field. For all problems we study, we parameterize advice by its size, i.e., we allow advice of a certain size $k$. Specifically, the advice need not necessarily encode the optimal solution or the request sequence itself. This opens up more possibilities to the algorithm designer in regards to the choice of an appropriate advice oracle, which may have further practical applications in machine learning.

\section{A warm-up: the ski rental problem}
\label{sec:ski.rental}


\subsection{Background}
The ski rental problem is a canonical example in online rent-or-buy problems. Here, the request sequence can be seen as vacation days, and on each
day the vacationer (that is, the algorithm) must decide whether to continue renting skis, or buy them. Without loss of generality we
assume that renting costs a unit per day, and buying costs $B \in \mathbb{N}^+$. The number of skiing days, which we denote by $D$,
is unknown to the algorithm, and we observe that the optimal offline cost is $\min\{D,B\}$.
Generalizations of ski rental have been applied in many settings, such as dynamic TCP acknowledgment~\cite{karlin2003dynamic},
the parking permit problem~\cite{Meyers05}, and snoopy caching~\cite{KaMaRS88}.

Consider the single-bit advice setting. Suppose that the advice encodes whether to buy on day 1, or always rent. An algorithm that blindly follows the
advice is optimal if the advice is trusted, but, if the advice is untrusted, the competitive ratio is as high as $D/B$, if $D>B$.
Hence, this algorithm is $(1, \infty)$-competitive, for $D \rightarrow \infty$.

\subsection{Ski rental with untrusted advice}
We define the family of algorithms $A_k$, with parameter $0<k\leq B$ as follows.
There is a single bit of advice, which is the indicator of the event $D<B$. If the advice bit is 1,
then $A_k$ rents until until day $B-1$ and buys on day $B$. Otherwise, the algorithm buys on day $k$.

\begin{proposition}
Algorithm $A_k$ is
$
(1+\frac{k-1}{B}, 1+\frac{B-1}{k})
$-competitive.
\label{prop:ski.adversarial}
\end{proposition}

\begin{proof}
	Table~\ref{table:cr} summarizes the competitive ratios of the algorithm, for the four different settings,
	depending on the value and the trustworthiness of the advice.
	\begin{table}[htb!]
		\begin{center}
			\begin{tabular}{l|ll}
				advice    & trusted & untrusted  \\ \hline \\[-.8em]
				{0 ($D \geq B$)}       & $\frac{k-1+B}{B}= 1+\frac{k-1}{B}$ & $\frac{k-1+B}{k}=1+\frac{B-1}{k}$ 
				\\[0.5em]
				{1 ($D <  B$)}     & 1 & $\frac{B-1+B}{B}=2-\frac{1}{B}$ \\
			\end{tabular}
			\caption{\label{table:cr} The competitive ratios of the family of algorithms $A_k$.}
		\end{center}
	\end{table}
	
	It follows that the trusted competitive ratio of $A_k$ is at most $1+\frac{k-1}{B}$ and its untrusted competitive
	ratio is at most $\max\{2-\frac{1}{B}, 1+\frac{B-1}{k}\} = 1+\frac{B-1}{k}$.
\end{proof}

%

It is worth noting that algorithm $A_k$ is slightly different from the one proposed in~\cite{NIPS2018_8174}. The latter buys on day $\lceil B^2/k \rceil$ if the advice is $1$  and on day $k$ if the advice is $0$, choosing $\lambda=k/B$ to match their notation, and in~\cite{NIPS2018_8174} it
is shown to be $(1+k/B,1+B/k)$-competitive.
More importantly, we show that $A_k$
is Pareto-optimal in the space of all deterministic online algorithms with advice of {\em any size}. This implies that more than
a single bit of advice will not improve the tradeoff between the trusted and untrusted competitive ratios.

\begin{theorem}\label{th:ski-main}
For any deterministic $(1+\frac{k-1}{B},w)$-competitive algorithm A, with $1 \leq k \leq B$, with advice of any size,
it holds that $w \geq 1+\frac{B-1}{k}$.
\label{thm:ski.one.suffices}
\end{theorem}

\begin{proof}
Let $A$ be an algorithm with trusted competitive ratio at most $1+\frac{k-1}B$.
First, note that if the advice is untrusted, the competitive ratio cannot be better than the competitive ratio of a purely online algorithm. For ski rental, it is known that no online algorithm can achieve a competitive ratio better than $1+(B-1)/B$~\cite{KaMaRS88}. So, in the case $k=B$, the claim trivially holds. In the remainder of the proof, we assume $k < B$.

We use $\sigma_D$ to denote the instance of the problem in which the number of skiing days is $D$, and use $A_t(\sigma_D)$ to denote the cost of $A$ for $\sigma_D$ in case of trusted advice.

Consider a situation  in which the input is $\sigma_{B+k}$ and the advice for $A$ is trusted. Let $j$ be the day the algorithm will buy under this advice.  Since the advice is trusted and thus $\Opt(\sigma_{B+k}) = B$, it must be that
\[
    A_t(\sigma_{B+k}) \leq \left(1+\frac{k-1}B\right) \Opt(\sigma_{B+k}),
\]
which implies $j<B+k$. In other words, $A$ indeed buys on day $j$. We conclude that $A_t(\sigma_{B+k}) = j-1+B$ which further implies $j\leq k$.

Let $x$ be the trusted advice $A$ receives on input $\sigma_{B+k}$ and suppose $A$ receives the same advice $x$ on input $\sigma_j$. Note that $x$ can be trusted or untrusted for $\sigma_j$.  The important point is that $A$ serves $\sigma_j$ in the same way it serves $\sigma_{B+k}$, that is, it rents for $j-1$ days and buys on day $j$. The cost of $A$ for $\sigma_j$ is then $ j-1 + B$, while $\Opt(\sigma_j) =j$. The ratio between the cost of the algorithm and the optimal cost is therefore $1+\frac{B-1}{j}$, which is at least $1+\frac{B-1}{k}$ since $j\leq k$.
Note that $1+\frac{B-1}{k} > 1 +\frac{k-1}{B}$ (since we assumed $k<B$) and therefore the advice in this situation has to be untrusted, by the assumption on the trusted competitive ratio of $A$. We conclude that the untrusted competitive ratio must be at least $1+(B-1)/k$.
\end{proof}

\section{Online bidding}
\label{sec:online.bidding}


\subsection{Background}
In the {\em online bidding} problem, a player wants to guess a hidden, unknown real value $u \geq 1$. To this end, the player submits a sequence
$X=(x_i)$ of increasing {\em bids}, until one of them is at least $u$. The strategy of the player is defined by this sequence of bids,
and the cost of guessing the hidden value $u$ is equal to $\sum_{i=1}^j {x_i}$, where $j$ is such that $x_{j-1}<u\leq x_j$. Hence the following natural
definition of the competitive ratio of the bidder's strategy.
\[
w_X= \sup_{u} \frac{\sum_{i=1}^j {x_i}}{u}, \textrm{ where $j$ is such that } x_{j-1}<u\leq x_j.
\]

The problem was introduced in~\cite{ChrKen06} as a canonical problem for formalizing doubling-based strategies in
online and offline optimization problems, such as searching for a target on the line, minimum latency, and hierarchical clustering.
It is worth noting that online bidding is identical to the problem
of minimizing the {\em acceleration ratio of interruptible algorithms}~\cite{RZ.1991.composing}; the latter and its
generalizations are problems with many practical applications in AI (see, for instance~\cite{aaai06:contracts}).

Without advice, the best competitive ratio is 4, and can be achieved using the doubling strategy $x_i=2^i$. If the advice encodes\footnote{We assume that the
advice provides the exact value $u$ to the algorithm. For practical considerations, it suffices to assume an oracle that provides an $(1+\epsilon)$-approximation of the hidden value, for sufficiently small $\epsilon>0$. This will only affect the competitive ratios by the same negligible factor.}
the value $u$,
and assuming trusted advice, bidding $x_1=u$ is a trivial optimal strategy.
The above  observations imply that there are simple strategies that are $(4,4)$-competitive and $(1,\infty)$-competitive, respectively.


\subsection{Online bidding with untrusted advice}

We will analyze the problem under two settings. In the first one (Section~\ref{subsec:bidding.target}), the advice encodes the target $u$, and is thus unlimited in size. In the second setting (Section~\ref{subsec:bidding.kbits}), the advice has fixed size $k$.

\subsubsection{Advice encodes the target $u$}
\label{subsec:bidding.target}

Suppose that $w \geq 4$ is a fixed, given parameter.
We will show a Pareto-optimal bidding strategy $X_u^*$, assuming that the advice encodes $u$,
which is $(\frac{w-\sqrt{w^2-4w}}{2},w)$-competitive (Theorem~\ref{thm:bidding.optimal} and Figure~\ref{fig:pareto-bidding}).

\begin{figure}
\begin{center}
\begin{tikzpicture}[thick,scale=1.5]
  \draw[->] (4,1 ) -- (8,1);
  \draw (8,.9) node[below] {$w$};
  \draw[->] (4,1) -- (4,3);
  \draw (3.9,3) node[left] {$r$};
  \draw[domain=4:8, samples=200, variable=\w, blue] plot ({\w},{(\w-(\w*\w-4*\w)^0.5)/2});
\foreach \y in {4,5,6,7} {
    \draw (\y,1)  -- (\y,0.9)
        node [below] {$\y$};
};
\foreach \x in {1,2} {
    \draw (4,\x)  -- (3.9,\x)
        node [left] {$\x$};
};
\end{tikzpicture}
\end{center}
\caption{The Pareto frontier of the online bidding problem, when the advice encodes the target $u$, as shown in Theorem~\ref{thm:bidding.optimal}.}
\label{fig:pareto-bidding}
\end{figure}

We begin with some definitions. Since the index of the bid which reveals the value will be important in the analysis, we define the class
$S_{m,u}$, with $m \in \mathbb{N}^+$ as the set of bidding strategies with advice $u$ which are $w$-competitive, and which, if the advice
is trusted, succeed in finding the value with precisely the $m$-th bid. We say that a strategy $X \in S_{m,u}$ that is
$(r,w)$-competitive dominates $S_{m,u}$ if for every $X' \in S_{m,u}$, such that $X'$ is $(r',w)$-competitive, $r\leq r'$ holds.

The high-level idea is to identify, for any given $m$, a dominant strategy in $S_{m,u}$. Let $X_{m,u}^*$ denote such a strategy, and denote
by $(r_{m,u}^*, w)$ its competitiveness. Then $X_{m,u}^*$ and $r_{m,u}^*$ are the solutions to an infinite linear program
which we denote by $P_{m,u}$, and which is shown below. For convenience, for any strategy $X$, we will always define $x_0$ to be equal to 1.

\hspace*{-1cm}
\begin{minipage}[t]{0.42\textwidth}
\begin{align*}
\min \quad &r_{m,u} \tag{$P_{m,u}$} \\
\textrm{s.t. } & x_i <x_{i+1}, \quad i \in \mathbb{N}^+ \\
& x_{m-1} <u\leq x_m \\
& \sum_{j=1}^m x_j \leq r_{m,u} \cdot u \\
& \sum_{j=1}^i x_j \leq  w \cdot x_{i-1}, \quad i \in \mathbb{N}^+ \\
& x_i \geq 0, \quad i \in \mathbb{N}^+.
\end{align*}
\end{minipage}
$\quad \quad$
\begin{minipage}[t]{0.5\textwidth}
\begin{align*}
\min \quad &\frac{1}{u} \cdot \sum_{i=1}^m x_i \tag{$L_{m,u}$} \\
\textrm{s.t. }
&x_i <x_{i+1}, \quad i \in \mathbb{N}^+ \\
& x_m = u \\
&\sum_{j=1}^i x_j \leq  w \cdot x_{i-1}, \quad i \in \mathbb{N}^+ \tag{$C_i$} \\
& x_i \geq 0, \quad i \in \mathbb{N}^+.
\end{align*}
\end{minipage}
\bigskip

Note that in $P_{m,u}$ the constraints  $\sum_{j=1}^i x_j \leq w \cdot x_{i-1}$ guarantee that the untrusted competitive ratio of $X$ is at most $w$,
whereas the constraints $\sum_{j=1}^m x_j \leq r_{m,u} \cdot u$ and $x_{m-1} <u\leq x_m$ guarantee that if the advice is trusted,
then $X$ succeeds in finding $u$ precisely with its $m$-th bid, and in this case the competitive ratio is $r_{m,u}$.

We also observe that an optimal solution $X^*_{m,u}=(x^*_{i})_{i\geq 1}$ for $P_{m,u}$ must be such that $x_{m}=u$,
otherwise one could define a strategy $X'_{m,u}$ in
which $x'_{i}=x^*_{i}/\alpha$, for all $i \geq 1$, with $\alpha=u/x^*_{m}$, which is still feasible for $P_{m,u}$,
is such that $x'_{m}=u$, and has better objective value than $X^*_{m,u}$, a contradiction. Furthermore, in an optimal solution, the constraint
$\sum_{i=1}^m x_i \leq r_{m,u} \cdot u$ must hold with equality. Therefore, $X^*_{m,u}$ and $r^*_{m,u}$ are also solutions to the  linear program
$L_{m,u}$.


Next, define
$r^*_u=\inf_m r^*_{m,u}$, and $r^* =\sup_u r^*_u$. Informally, $r^*_u$, $r^*$ are the optimal competitive ratios,
assuming trusted advice. More precisely, the dominant strategy in the space of all $w$-competitive strategies is
$(r^*_u,w)$-competitive, and $r^*$ is an upper bound on $r^*_u$, assuming the worst-case choice of $u$.

We first argue how to compute $r^*_{m,u}$ and the corresponding strategy $X_{m,u}^*$, provided that $L_{m,u}$ is feasible.
This is accomplished in Lemma~\ref{lem:bidding.recurrence}. 
The main idea behind the technical proof is to show that
in an optimal solution of $L_{m,u}$, all constraints $C_i$ hold with equality. This allows us to
describe the bids of the optimal strategy by means of a linear recurrence relation
which we can solve so as to obtain an expression for the bids of $X^*_{m,u}$. We refer the reader to the appendix, for the technical details. 

To introduce the next lemma, we need to define sequences $a_i$ and $b_i$ as follows:
\begin{equation}
  \label{eq:rec_a'}
  a_i = \frac{a_{i-1}}{w-1-b_{i-1}}, \textrm{ with } a_0 = 1,
\end{equation}

\begin{equation}
  \label{eq:rec_b}
  b_i = \frac{1 + b_{i-1}}{ w - 1 - b_{i-1}}, \textrm{ with } b_0 = 0,
\end{equation}
In addition, denote the roots of $x^2 - wx + w$ by
$$
	\rho_1=\frac{w-\sqrt{w^2-4w}}{2} \mbox{ and } \rho_2=\frac{w+\sqrt{w^2-4w}}{2}.
$$
\begin{lemma}[Appendix] 	\label{lem:bidding.recurrence}
For every $m$ define $X_{m,u}$ as follows:
\begin{itemize}
\item If $w>4$, then $x_{m,u,i}=\alpha \cdot \rho_1^{i-1} +\beta \cdot \rho_2^{i-1}$, where $\alpha=\frac{a_{m-1}\rho_2^{m-1}-1}{\rho_2^{m-1}-\rho_1^{m-1}} \cdot u$,
and $\beta=\frac{a_{m-1}\rho_1^{m-1}-1}{\rho_1^{m-1}-\rho_2^{m-1}} \cdot u$,
\item If $w=4$, then $x_{m,u,i}=(\alpha+\beta \cdot i)\cdot 2^i$, where $\alpha=\frac{2^{m-1}\cdot m \cdot a_{m-1}-1}{2^m(m-1)} \cdot u$, and $\beta=\frac{1-2^{m-1}\cdot a_{m-1}}{2^m(m-1)} \cdot u$.
\end{itemize}
Then, $X_{m,u}$ is an optimal feasible solution if and only if $a_{m-1}\cdot u \leq w$.
\end{lemma}

We can now give the statement of the optimal strategy $X_u^*$. First, we can argue that the optimal objective value of
$L_{m,u}$ is monotone increasing in $m$, thus it suffices to find the objective value of the smallest $m^*$ for which $L_{m^*,u}$ is feasible;
This can be accomplished with a binary search in the interval $[1,\lceil \log u \rceil]$ , since we know that the doubling strategy in
which the $i$-th bid equals $2^i$ is $w$-competitive for all $w \geq 4$; hence $m^* \leq \lceil \log u \rceil $.
Then $X^*_u$ is derived as in the statement of Lemma~\ref{lem:bidding.recurrence}. 

Last, the following lemma allows us to express $r^*$ as a function of the values of the sequence $b$, which we can further exploit so as to obtain the exact value of $r^*_u$.
\begin{lemma}[Appendix]
It holds that $r^*=1+\sum_{i=1}^\infty \prod_{j=1}^{i-1} b_j$. Furthermore, $r^*=\frac{w-\sqrt{w^2-4w}}{2}$.
\label{lem:bidding.r}
\end{lemma}

Combining Lemmas~\ref{lem:bidding.recurrence} and \ref{lem:bidding.r} we obtain our main result of this section:
\begin{theorem}
Strategy $X^*_u$ is Pareto-optimal and is $(\frac{w-\sqrt{w^2-4w}}{2},w)$-competitive.
\label{thm:bidding.optimal}
\end{theorem}

\subsubsection{Advice is a $k$-bit string}
\label{subsec:bidding.kbits}

The strategy $X_u^*$ studied in the previous section requires $u$ as advice, which can be unbounded.
A natural question is what competitiveness can one achieve with $k$ advice bits, for some fixed $k$. We address this question
both from the point of view of upper and lower bounds. Concerning upper bounds, we show the following:

\begin{theorem}
  For every $w\geq 4$,
  there exists a bidding strategy with $k$ bits of
  advice which is $(r, w)$-competitive, where
  \[
      r = \left\{
      \begin{array}{ll}
        \frac{\left(w+\sqrt{w^2-4w}\right)^{1+1/K}}{2^{1/K}(w+\sqrt{w^2-4w}-2)}
        & \mbox{if } w \leq (1+K)^2/K
        \\[1em]
        \frac{(1+K)^{1+1/K}}{K}
        &
        \mbox{if }  w \geq (1+K)^2/K.
      \end{array}
      \right.
  \]
  and  where $K=2^k$.
\label{thm:bidding.upper.k}
\end{theorem}


\begin{proof}
	The strategy has a parameter $\rho>1$ which will be specified later.
	The advice corresponding to the hidden value $u$ is the number $a \in \{ 0, \ldots, K-1 \}$ such that there exists $i \in \mathbb{N}$ for which $\rho^{i+\frac{a}{K}}\geq u$ and $\rho^{i+\frac{a}{K}}$ is as small as possible. From this we derive that 
	\[
a=\lceil K \log_{\rho} u \rceil \bmod K.
	\]
For the intuition behind the above choice, we may think of $K$ different bidding sequences of the form 
$(\rho^{i+\frac{l}{K}})_{i\geq 0}$, with $l\in \{0,\ldots k-1\}$, then the advice points to the sequence which discovers the target $u$ at minimum possible cost.


	If the advice is untrusted, then the worst-case untrusted competitive ratio occurs when $u$ is infinitesimally larger than some $x_{i-1}$. Thus,
	\[
	w_A= \sup_i \frac{\sum_{j=1}^i x_j}{x_{i-1}}=  \sup_i \frac{\sum_{j=1}^i x_j}{x_i/\rho} = \frac{\rho^2}{\rho-1}.
	\]
	
	This bound is convex in $\rho$ and equals $w$ for the following values
	\[
	\rho_{1,2} = \frac{w \mp \sqrt{w^2 - 4w}}{2}.
	\]
	Hence for any $\rho\in[\rho_1,\rho_2]$ the untrusted ratio is at most $w$.
	
	To analyze the trusted competitive ratio $r_A$ we observe that if the advice is correct, the ratio between the hidden value and the first
	successful bid is most $\rho^{1/K}$.  Thus the trusted ratio is at most
	\begin{align*}
		r_A \leq \sup_i \frac{\sum_{j=1}^i x_j}{x_i/\rho^{1/K}}
		= \frac{\rho^{1+1/K}}{\rho-1}.
	\end{align*}
	We now argue how to choose an appropriate base $\rho$.
	Using second order analysis we observe the above upper bound on $r_A$ has an extreme point at $\rho=1+K$.
	The first derivative of the ratio is then
	\[
	\frac{\rho^{1/K} (-K+\rho-1)}{K (\rho-1)^2}
	\]
	which is negative for $1< \rho < 1+K$ and positive for $\rho > 1+K$. Hence the ratio is minimum at its extreme point $\rho=1+K$.
	
	However, the choice of $\rho$ needs to satisfy $\rho\in[\rho_1,\rho_2]$. Since $\rho_2$ is increasing in $w$, there is a threshold $w_0$ such that $1+K \leq \rho_2$ iff $w\geq w_0$. This value is
	\[
	w_0 = \frac{(1+K)^2}{K}.
	\]
	Choosing $\rho=\rho_2$ whenever $w\leq w_0$ and $\rho=1+K$ yields the claimed bound on the right ratio $r_A$.
	Note that for both choices we have $\rho\geq \rho_1$ as required since $\rho_1\leq 2$ because $\rho_1$ is decreasing in $w$ and equals $2$ at $w=4$.
\end{proof}

In particular, for $w=4$, the strategy of Theorem~\ref{thm:bidding.upper.k} is $(2^{1+\frac{1}{2^k}},4)$-competitive, whereas $X_u^*$ is $(2,4)$-competitive.
In what follows, we will prove 
a lower bound on the competitiveness of any bidding strategy with $k$ bits (Theorem~\ref{thm:bidding.lower.k}). The result shows that one needs an unbounded
number of bids to achieve $(2,4)$-competitiveness. 

We begin with a useful property of all competitively optimal strategies (in the standard, no-advice model). The property
essentially shows that competitively optimal strategies are small variants of the doubling strategy.

\begin{lemma}
	For any 4-competitive strategy $X=(x_i)$ for online bidding, it holds that
	\[
	x_i \leq (2+\frac{2}{i})x_{i-1},
	\]
	for all $i \geq 1$, where $x_0$ is defined to be equal to 1.
	\label{lemma:bb}
\end{lemma}
\begin{proof}
	For $i=1$, the claim holds since $x_1$ must be at most 4, otherwise
	the strategy cannot be 4-competitive. We can then assume that $i \geq 2$,
	and prove the claim by induction on $i$.
	Suppose that the claim holds for all $j \leq i$, that is
	$x_{j} \leq (2+\frac{2}{j})x_{j-1}$, for all $j \leq i$. This implies
	that
	\begin{equation}
		x_{i-j} \geq \frac{1}{\prod_{k=0}^{j-1} (2+\frac{2}{i-k})}x_i.
		\label{eq:bb.1}
	\end{equation}
	We will show that the claim holds for $i+1$. From the 4-competitiveness of the strategy $X$ we have that
	\[
	\frac{\sum_{j=1}^{i+1} x_j}{x_i} \leq 4 \Rightarrow x_{i+1} +\sum_{k=1}^{i-1} x_k \leq 3x_i,
	\]
	and substituting $x_1, ,\ldots x_{i-1}$ using~\eqref{eq:bb.1}, we obtain that
	\[
	x_{i+1} \leq (3-P_i)x_i, \quad \mbox{ where } P_i=\sum_{k=1}^{i-1} \frac{1}{\prod_{j=0}^{k-1} (2+\frac{2}{i-j})}.
	\]
	It then suffices to show that
	\begin{equation}
		3-P_i \leq 2+\frac{2}{i+1} \quad \mbox{or equivalently } \ P_i \geq 1-\frac{2}{i+1}.
		\label{eq:bb.2}
	\end{equation}
	We will prove~\eqref{eq:bb.2} by induction on $i$, in fact, even stronger, we will show that~\eqref{eq:bb.2} holds with
	equality. For $i=2$, this claim can be readily verified. Assuming that it holds
	for $i$, we will show that it holds for $i+1$.
	Indeed we have that
	\[
	P_{i+1} =\frac{1}{2+\frac{2}{i+1}} (1+P_i) = \frac{1}{2+\frac{2}{i+1}} \left(2-\frac{2}{i+1}\right)
	= 1-\frac{2}{i+2},
	\]
	where the second equality follows from the induction hypothesis, and the third equality can be readily verified.
\end{proof}

\begin{corollary}
	For any 4-competitive strategy $X=(x_i)$, and every $\epsilon>0$, there exists $i_0$ such that for all $i\geq i_0$, it holds
	that
	\[
	\sum_{j=1}^{i} x_j \geq (2-\epsilon)x_i.
	\]
	\label{cor:bb}
\end{corollary}
\begin{proof}
	It holds that
	\[
	\sum_{j=1}^{i} x_j=x_i +\sum_{j=1}^{i-1} x_j \geq x_i\left(1+\sum_{k=1}^{i-2} \frac{1}{\prod_{j=0}^{k-1} (2+\frac{2}{i-j})}\right)=
	x_i(1+P_{i-1}) \geq x_i\left(2-\frac{2}{i}\right),
	\]
	where the inequality follows from Lemma~\ref{lemma:bb}, and the last inequality holds from the property on $P_i$ that was
	shown in the proof of Lemma~\ref{lemma:bb}. Last, note that $\frac{2}{i} \leq \epsilon$, for sufficiently large $i$.
\end{proof}

\begin{theorem}
For any bidding strategy with $k$ advice bits that is $(r,4)$-competitive it holds that $r \geq 2+\frac{1}{3\cdot 2^k}$.
\label{thm:bidding.lower.k}
\end{theorem}

\begin{proof}
	With $k$ bits of advice, the online algorithm can only differentiate between $K=2^k$ online bidding sequences, which
	we denote by $X_j=(x_{j,i})_{i\geq 1}$, with $j \in [1,K]$. (In words, the $i$-th bid in the $j$-th strategy is equal to
	$x_{j,i}$.) That is, the advice is a function that maps the hidden value $u$ to one of these $K$ strategies.
	We will say that the advice {\em chooses $X_j$}, or simply {\em chooses $j$}, with $j \in [1,K]$
	to describe this action. Define
	\[
	\delta=\frac{1}{3K},
	\]
	and let $r$ denote the trusted competitive ratio of a given strategy with $k$ advice bits.
	We will prove that $r> 2+\delta$. By way of contradiction, we will assume that $r \leq 2+\delta$,
	and we will prove that there exists a choice of the hidden value which results in a trusted competitive ratio larger than
	$2+\delta$. Note that each $X_j$, with $j \in [1,K]$ must be 4-competitive, since the untrusted competitive ratio
	has to be 4.

	For a given index $i$, and for any $j \in [1,K-1]$, let $l_{j,i}$ denote the smallest index such that $x_{j, l_{j,i}}>x_{K,i-1}$.
	To simplify notation, we will define $\overline{x}_j \equiv x_{j, l_{j,i}}$. Such $l_{j,i}$ must always
	exist; moreover, since $i$ can be unbounded, so is $x_{K,i-1}$, and hence also $l_{j,i}$. Thus we can
	invoke Corollary~\ref{cor:bb}, and obtain that for any $\epsilon>0$, there exists $i_0$ such that
	\[
	T_j(\overline{x_j}) \geq (2-\epsilon) \overline{x}_j, \quad \mbox{ for all $j \in [1,K-1]$,}
	\]
	where we use the notation $T_j(x)$ to denote the sum of all bids less than or equal to $y$ in
	$X_j$, i.e., $\sum_{y \in X_j, y \leq x}y$.
	
	We will assume, without loss of generality, that $\overline{x}_1 \leq \overline{x}_2 \ldots \leq \overline{x}_{K-1}$.

	To prove the theorem, we will use a game between the online algorithm with advice, and an adversary, which proceeds in rounds.
	In each round, the adversary will consider a hidden value $u$ from the following candidate set:
	\[
	\{x_{K,i-1}+\epsilon, \overline{x}_1+\epsilon, \ldots \overline{x}_{K-1}+\epsilon \},
	\]
	where $\epsilon \rightarrow 0$ is an infinitesimally small value. These values are intuitive: The first value is a ``bad'' choice for $X_K$,
	and the remaining values are ``bad'' choices for $X_1, \ldots X_{K-1}$, respectively. For each choice, the online algorithm needs to maintain
	a trusted competitive ratio of at most $2+\delta$. To do so, we will argue that the algorithm will need to choose sequences of ``progressively
	larger indexes''. By this we mean that if in some round the advice chooses sequences $X_j$, then when presented with
	$u=\overline{x}_j+\epsilon$ in the next round,
	the advice will have to choose $X_{j'}$, with $j'>j$. Moreover, we will argue that the $\overline{x}_j$'s used
	as the hidden value by the adversary cannot be  too big in comparison to $x_{K,i-1}$, if the algorithm is to be $(2+\delta)$-competitive.
	At the end, the ``top'' strategy, namely strategy $X_K$, will be the best choice. Then $r$ will be at least
	\[
	\frac{\sum_{j=1}^i x_{K,i}}{\overline{x}_j},
	\]
	where $\overline{x}_j$ is not too big in comparison to $x_{K,i-1}$, i.e., by a factor of at most $e^{1/3}$. This will imply that the trusted competitive ratio
	must be at least $1/e^{1/3}$ times 4 (since 4 is the untrusted competitive ratio of $X_K$), which is larger than $2+\delta$ for all $K$, a contradiction.
	
	\paragraph{Description of the game}
	The adversarial game proceeds in rounds. In round $0$, the adversary chooses $u$ to be infinitesimally larger than $x_{K,i-1}$.

	There are the following subcases:

	\noindent
	\begin{enumerate}
		\item[\bf 1]
		{\em The algorithm chooses some $j \in [1,K-1]$.} In this case it must be that
		\[
		r \geq \frac{T_j(\overline{x}_j)}{x_{K,i-1}} \geq \frac{(2-\epsilon)\overline{x}_j}{x_{K,i-1}},
		\]
		and since $r=2+\delta$, we obtain that
		\begin{equation}
			\overline{x}_j \leq \frac{2+\delta}{2-\epsilon} x_{K,i-1}.
			\label{eq:bid.rounds}
		\end{equation}
		
		\item[ \bf 2]
		{\em The algorithm chooses $K$}. In this case we have that
		\[
		r\geq \frac{T_K (x_{K,i})}{x_{K,i-1}}.
		\]
	\end{enumerate}
	
	The remaining rounds are defined inductively. Suppose that rounds $0, \ldots ,\rho$ have been defined. We will now define round $\rho+1$.
	If in round $\rho$ the advice chose $X_K$, the game is over, and no actions take place in any subsequent rounds. Otherwise, let $j_{\rho}$ be such that
	the advice chose $X_{j_\rho}$ in round $\rho$. For round $\rho+1$, the adversary chooses $u$ to be infinitesimally larger than $\overline{x}_{j_\rho}$.
	We consider the following cases concerning the algorithm.
	\begin{enumerate}
		\item [\bf 1]
		{\em The algorithm chooses some $j=j_{\rho+1}>j_\rho$, with $j<K$.} In this case there must be that
		\[
		r \geq \frac{T_j(\overline{x}_j)}{\overline{x}_{j_\rho}} \geq \frac{(2-\epsilon)\overline{x}_j}{\overline{x}_{j_\rho}},
		\]
		and since $r=2+\delta$, we obtain, that
		\[
		\overline{x}_j \leq \frac{2+\delta}{2-\epsilon} \overline{x}_{j_\rho}.
		\]
		
		\item [\bf 2]
		{\em The algorithm chooses $X_K$}. In this case we have that
		\[
		r\geq \frac{T_K(x_{K,i})}{\overline{x}_{j\rho}}.
		\]
		
		\item [\bf 3]
		{\em The algorithm chooses $j$ with $j \leq j_\rho$}. We will argue that this case cannot occur, unless $r>2+\delta$.
	\end{enumerate}

	We prove the claim concerning the last case. First, suppose that the algorithm chose $j_{\rho+1} =j_\rho$. Then we have that
	\[
	r \geq \frac{T_j(\overline{x}_{j_\rho})+\overline{x}_{j_\rho}}{\overline{x}_{j_\rho}} \geq 3-\epsilon > 2+\delta.
	\]
	Suppose now that the algorithm chose $j_{\rho+1}=j< j_\rho$. Then we have that
	\[
	r\geq \frac{ T_j(\overline{x}_j)+\overline{x}_{j_\rho} } {\overline{x}_{j_\rho}} \geq \frac{(2-\epsilon)\overline{x}_j+\overline{x}_{j_\rho}}{\overline{x}_{j_\rho}}
	\geq 1+\frac{(2-\epsilon)x_{K,i-1}}{\overline{x}_{j_\rho}}.
	\]
	Since it must be that $r \leq 2+\delta$, it follows that
	\[
	\overline{x}_{j_\rho} \geq \frac{2-\epsilon}{1+\delta} x_{K,i-1}.
	\]
	However, from~\eqref{eq:bid.rounds}, we know that $\overline{x}_{j_\rho}$ cannot exceed $(\frac{2+\delta}{2-\epsilon})^K \cdot x_{K,i-1}$, and since $\epsilon$ can be arbitrarily small, we have that
	\[
	\overline{x}_{j_\rho} \leq (1+\frac{1}{3K})^K \cdot x_{K,i-1}.
	\]
	It then follows that
	\[
	(1+\frac{1}{3K})^K \geq \frac{2}{1+\frac{1}{3K}},
	\]
	which is a contradiction, as it can readily be confirmed using standard calculus.
	
	\paragraph{Analysis of the game}
	Let $R$ denote the last round of the game. From the arguments above it follows that either $r \geq 3-\epsilon$,
	or in round $R$ the advice chose algorithm $X_K$. As argued above, in this case we have that
	\[
	\overline{x}_{j_R} \leq (1+\frac{1}{3K})^K x_{K,i-1} \leq e^\frac{1}{3} x_{K,i-1}.
	\]
	Thus, from the corresponding action in round $R$ it must be that
	\[
	r \geq \frac{T_K(x_{K,i})}{\overline{x}_{j_R}} \geq \frac{T_K(x_{K,i})}{e^\frac{1}{3} x_{K,i-1}}.
	\]
	Note that, as discussed early in the proof, the above holds for all sufficiently large $i\geq i_0$. Therefore,
	\[
	r \geq \sup_{i \geq {i_0}} \frac{\sum_{j=1}^i x_{K,j}} {e^\frac{1}{3} x_{K,i-1}} =
	\frac{1}{e^\frac{1}{3}} \sup_{i=1} \frac{\sum_{j=1}^i x_{K,j}} {x_{K,i-1}} \geq \frac{4}{e^\frac{1}{3}} >2+\frac{1}{3K}.
	\]
	We conclude that  $r \geq 2+\frac{1}{3K}$.

\end{proof}


\newcommand{\RRC}{\ensuremath{\textsc{Rrc}}\xspace}
\newcommand{\sh}[1]{}
\newcommand{\neww}[1]{\textcolor{red}{#1}}

\section{Online bin packing}
\label{sect:bp}

\subsection{Background}
{Online bin packing finds applications in a broad range of practical problems, from server consolidation to cutting stock problems. We refer the reader to a survey by Coffman et al.~\cite{BPsurvey2013} and a brief introduction by Johnson~\cite{Johnson16} for details on bin packing and its applications.  Along with its practical significance, research on this problem has lead to technical developments for online algorithms in general.
}

An instance of the online bin packing problem consists of a sequence of items with different \emph{sizes} in the range $(0,1]$, and the objective is to pack these items into a minimum number of {bins}, each with a capacity of 1. For each arriving item, the algorithm must place it in one of the current bins or open a new bin for the item. {By the nature of the bin packing problem, it cannot be avoided that a constant number of bins are not well filled.  Hence it is standard practice to measure the performance using the \emph{asymptotic competitive ratio}. Formally, w}e say that algorithm $A$ has an asymptotic competitive ratio $r$ if, on every sequence $\sigma$, the number of opened bins satisfies $A(\sigma) \leq r \cdot \opt(\sigma) + c$, where $c$ is a constant. \sh{This ratio contrasts with the so-called \emph{absolute competitive ratio} for which the constant $c$ has to be zero. }As standard in the analysis of bin packing problems, throughout this section, by ``competitive ratio'' we mean ``asymptotic competitive ratio''.

{The 
most practical algorithms 
for bin packing are First-Fit and Best-Fit. First-Fit
}
maintains bins in the same order that they have been opened, and places an item into the first bin with enough free space; if no such bin exists, it opens a new bin.
Best-Fit works similarly, except that it maintains bins in the non-increasing order of their \emph{level}, where level of a bin is the total size of its items.
{Johnson~\cite{Johnso73} proved that the competitive ratio 
of 
First-Fit and Best-Fit
is 1.7. Many other algorithms with improved competitive ratios have been studied. The best known algorithm was introduced by Balogh et al.~\cite{BaloghBDEL18} and has a competitive ratio of at most 1.5783. Moreover, it is known that no online algorithm can achieve a competitive ratio better than 1.54278~\cite{BaloghArX05554}.}

Online bin packing has also been studied in the advice setting~\cite{BoyarKLL16,RenaultRS15,ADKRR18}. In particular,\sh{a result by Angelopoulos et al.~\cite{ADKRR18} shows that with only constant number of bits,} it is possible to achieve a competitive ratio of 1.4702 with only a constant number of (trusted) advice bits~\cite{ADKRR18}. A restricted version of the bin packing problem, where items take sizes from a discrete set $\{1/k,2/k,\ldots, 1\}$, for some constant integer $k$, has been recently studied under a prediction model~\cite{DBLP:journals/jair/AngelopoulosKS23}, where it is assumed that predictions concern frequency of each item size. We note that this type of prediction is not applicable to the general setting of the bin packing problem.

In this section, we introduce an algorithm named Robust-Reserve-Critical
(\textsc{Rrc})
which has a parameter $\alpha \in [0,1]$, indicating how much the algorithm relies on the advice. Provided with $O(1)$ bits of advice, the algorithm is asymptotically $(r_\textsc{Rrc},w_\textsc{Rrc})$-competitive for 
$r_\textsc{Rrc} = 1.5 + \frac{3-3\alpha}{12-8\alpha}$ 
and 
$w_\textsc{Rrc} = 1.5+\max\{\frac{1}{4},\frac{6\alpha}{6-4\alpha}\}$.
If the advice is reliable, we set $\alpha=1$ and the algorithm is asymptotically 
$(1.5,4.5)$-competitive; otherwise, we set $\alpha$ to a smaller value. For $\alpha=0.9$ for example, the algorithm is asymptotically 
$(1.5625,3.75)$-competitive. Similarly, for $\alpha=0.868$, we get a $(1.5783,3.56)$-algorithm, whose consistency is the same as best existing online algorithms. Therefore, our results are useful for values of $\alpha > 0.868$, where improvements over online algorithms without prediction are realized. 
In contrast, for $\alpha <0.868$, the best-known competitive algorithms without predictions dominate our proposed solution. 

\subsection{The Reserve-Critical algorithm}
Our solution uses an algorithm introduced by Boyar et al.~\cite{BoyarKLL16} which achieves a competitive ratio of 1.5 using $O(\log n)$ bits of advice. We refer to this algorithm as Reserve-Critical in this paper and describe it briefly. See Figure~\ref{fig:resCrit} for an illustration.
The algorithm classifies items according to their size.  Tiny items have their size in the range $(0,1/3]$, small items in $(1/3,1/2]$, critical items in $(1/2,2/3]$, and large items in $(2/3,1]$. In addition, the algorithm has four kinds of bins, called tiny, small, critical and large bins. Large items are placed alone in large bins, which are opened at each arrival. Small items are placed in pairs in small bins, which are opened every other arrival. Critical bins contain a single critical item, and tiny items up to a total size of $1/3$ per bin, while tiny bins contain only tiny items.  The algorithm receives as advice the number of critical items, denoted by $c$, and opens $c$ \emph{critical bins} at the beginning.  Inside each critical bin, a space of 2/3 is reserved for a critical item, and tiny items are placed using First-Fit into the remaining space of these bins possibly opening new bins dedicated to tiny items.  Each critical item is placed in one of the critical bins.   
Note that the algorithm is heavily dependent on the advice being trusted. Imagine that the encoded advice 
overestimates the
number of critical items. This results in critical bins which contain only tiny items.
\sh{By the First-Fit strategy, all of them, except possibly for a single bin, are guaranteed to be filled up to a level of $1/6$ at least.  }The worst case is reached when tiny items form a subsequence $(1/6,\epsilon, 1/6, \epsilon, \ldots)$, 
while there is no critical item. In this case, all critical bins are filled up to a level slightly more than $1/6$.  Hence, untrusted advice can result in a competitive ratio as bad as 6.

\begin{figure}[htb]
	\tikzset{every picture/.style={line width=0.75pt}} 
\begin{center}  
\begin{tikzpicture}[x=0.75pt,y=0.75pt,yscale=-1.2,xscale=1.2]

\draw  [draw opacity=0][fill=black!100,fill opacity=1 ] (90,110) -- (110,110) -- (110,151) -- (90,151) -- cycle ;
\draw  [draw opacity=0][fill=black!75,fill opacity=1 ] (90,170) -- (110,170) -- (110,152) -- (90,152) -- cycle ;
\draw  [draw opacity=0][fill=black!50,fill opacity=1 ] (90,170) -- (110,170) -- (110,191) -- (90,191) -- cycle ;
\draw  [draw opacity=0][fill=black!25,fill opacity=1 ] (90,191) -- (110,191) -- (110,230) -- (90,230) -- cycle ;
\draw  [draw opacity=0][fill=black  ,fill opacity=1 ] (351,127) -- (371,127) -- (371,229.97) -- (351,229.97) -- cycle ;
\draw   (351,110.97) -- (371,110.97) -- (371,229.97) -- (351,229.97) -- cycle ;
\draw  [color={rgb, 255:red, 255; green, 255; blue, 255 }  ,draw opacity=1 ][fill=black!25  ,fill opacity=1 ] (300,154.97) -- (320,154.97) -- (320,160.97) -- (300,160.97) -- cycle ;
\draw  [draw opacity=0][fill=black!75  ,fill opacity=1 ] (300,160.97) -- (320,160.97) -- (320,230.97) -- (300,230.97) -- cycle ;
\draw  [color={rgb, 255:red, 255; green, 255; blue, 255 }  ,draw opacity=1 ][fill=black!25  ,fill opacity=1 ] (300,146.97) -- (320,146.97) -- (320,154.97) -- (300,154.97) -- cycle ;
\draw  [color={rgb, 255:red, 255; green, 255; blue, 255 }  ,draw opacity=1 ][fill=black!25  ,fill opacity=1 ] (300,136) -- (320,136) -- (320,146.97) -- (300,146.97) -- cycle ;
\draw   (300,110) -- (320,110) -- (320,230.97) -- (300,230.97) -- cycle ;
\draw  [color={rgb, 255:red, 255; green, 255; blue, 255 }  ,draw opacity=1 ][fill=black!25  ,fill opacity=1 ] (200,223.97) -- (220,223.97) -- (220,229.97) -- (200,229.97) -- cycle ;
\draw  [color={rgb, 255:red, 255; green, 255; blue, 255 }  ,draw opacity=1 ][fill=black!25  ,fill opacity=1 ] (200,215.97) -- (220,215.97) -- (220,223.97) -- (200,223.97) -- cycle ;
\draw  [color={rgb, 255:red, 255; green, 255; blue, 255 }  ,draw opacity=1 ][fill=black!25  ,fill opacity=1 ] (200,205) -- (220,205) -- (220,215.97) -- (200,215.97) -- cycle ;
\draw  [color={rgb, 255:red, 255; green, 255; blue, 255 }  ,draw opacity=1 ][fill=black!25  ,fill opacity=1 ] (200,200.97) -- (220,200.97) -- (220,206.97) -- (200,206.97) -- cycle ;
\draw  [color={rgb, 255:red, 255; green, 255; blue, 255 }  ,draw opacity=1 ][fill=black!25  ,fill opacity=1 ] (200,192.97) -- (220,192.97) -- (220,200.97) -- (200,200.97) -- cycle ;
\draw  [color={rgb, 255:red, 255; green, 255; blue, 255 }  ,draw opacity=1 ][fill=black!25  ,fill opacity=1 ] (200,182) -- (220,182) -- (220,192.97) -- (200,192.97) -- cycle ;
\draw   (200,110) -- (220,110) -- (220,229.97) -- (200,229.97) -- cycle ;
\draw  [color={rgb, 255:red, 255; green, 255; blue, 255 }  ,draw opacity=1 ][fill=black!50  ,fill opacity=1 ] (250,186) -- (270,186) -- (270,230) -- (250,230) -- cycle ;
\draw  [color={rgb, 255:red, 255; green, 255; blue, 255 }  ,draw opacity=1 ][fill=black!50  ,fill opacity=1 ] (250,147) -- (270,147) -- (270,186) -- (250,186) -- cycle ;
\draw   (250,109.85) -- (270,109.85) -- (270,230) -- (250,230) -- cycle ;

\draw (68,83) node [anchor=north west][inner sep=0.75pt]   [align=left] {\textbf{item classes}};
\draw (113,201) node [anchor=north west][inner sep=0.75pt]   [align=left] {tiny};
\draw (113,170) node [anchor=north west][inner sep=0.75pt]   [align=left] {small};
\draw (113,151) node [anchor=north west][inner sep=0.75pt]   [align=left] {critical};
\draw (113,118) node [anchor=north west][inner sep=0.75pt]   [align=left] {large};
\draw (196,83) node [anchor=north west][inner sep=0.75pt]   [align=left] {\textbf{bin classes}};
\draw (197,232) node [anchor=north west][inner sep=0.75pt]   [align=left] {tiny};
\draw (244,232) node [anchor=north west][inner sep=0.75pt]   [align=left] {small};
\draw (288,232) node [anchor=north west][inner sep=0.75pt]   [align=left] {critical};
\draw (346,232) node [anchor=north west][inner sep=0.75pt]   [align=left] {large};
\draw (322,142) node [anchor=north west][inner sep=0.75pt]  [font=\footnotesize] [align=left] {$\leq \frac{1}{3}$};
\draw    (100,110) -- (100,230) ;
\draw    (90,110) node[left] {1} -- (110,110);
\draw    (90,230)  node[left] {0}-- (110,230) ;
\draw    (90,151)  node[left] {2/3} -- (110,151) ;
\draw    (90,191)  node[left] {1/3} -- (110,191) ;
\draw    (90,170) node[left] {1/2} -- (110,170) ;
\end{tikzpicture}
\end{center}
\caption{Item and bin classification according to Boyar et al.~\cite{BoyarKLL16}.\label{fig:resCrit}}
\end{figure}

\sh{Let $t$ be the number of tiny bins opened by the algorithm. Equivalently the advice for the algorithm could be the fraction $c/(c+t)$ instead of the number $c$. We call this fraction the \emph{critical ratio}. Then the algorithm would open critical and tiny bins as needed, maintaining a proportion between them close to the given critical ratio. The precise mechanism is explained in the next section.
In a variant of Reserve-Critical the critical ratio is given to the algorithm only up to a precision of $k$ bits, and its competitive ratio has been analyzed as a function of $k$.  This variant is introduced and analyzed in \cite{ADKRR18}.}


\subsection{The Robust-Reserve-Critical (RRC) algorithm}
Let $t$ be the number of tiny bins opened by the Reserved-Critical algorithm. Recall that $c$ is the number of critical bins. We call the fraction $c/(c+t)$ the \emph{critical ratio}.
The advice for \textsc{Rrc} is a fraction $\gamma$, integer multiple of $1/2^{k}$, that is encoded in $k$ bits such that if the advice is trusted then 
$\gamma \leq c/(c+t) \leq \gamma + 1/2^k$. In case $c/(c+t)$ is a positive integer multiple of $1/2^k$, we break the tie towards $\gamma<c/(c+t)$.
Note that for a sufficiently large, yet constant, number of bits, $\gamma$ provides a good approximation of the critical ratio. {Indeed having $\gamma$ as advice is sufficient to achieve a competitive ratio that approaches $1.5$ in the trusted advice model, as shown in \cite{ADKRR18}.}

The \textsc{Rrc} algorithm has a parameter 
$0\leq \alpha \leq 1$, which together with the advice $\gamma$ can be used to define a fraction $\beta = \min \{ \alpha, \gamma\}$.  The algorithm maintains a proportion close to $\beta$ of critical bins among critical and tiny bins.
Formally, on the arrival of a critical item, the algorithm places it in a critical bin, opening a new one if necessary.  Each arriving tiny item $x$ is packed 
in the first critical bin which has enough space, with the restriction that the tiny items do not exceed a fraction 1/3 in these bins. If this fails, the algorithm tries to pack $x$ in a tiny bin using First-Fit strategy (this time on tiny bins). If this fails as well, a new bin $B$ is opened for $x$. Now, $B$ should be declared as a critical or a tiny bin. Let $c'$ and $t'$ denote the number of critical and tiny bins before opening $B$. If $c'+t'>0$ and $\frac{c'}{c'+t'} < \beta$, then $B$ is declared a critical bin; otherwise, $B$ is declared a tiny bin. Large and small items are placed similarly to the Reserved-Critical algorithm (one large item in each large bin and two small items in each small bin).

\subsection{Analysis}
Intuitively, \RRC works similarly to Reserved-Critical except that it might not open as many critical bins as suggested by the advice. The algorithm is more ``conservative'' in the sense that it does not keep two thirds of many (critical) bins open for critical items that might never arrive. The smaller the value of $\alpha$ is, the more conservative the algorithm is. Our analysis is based on two possibilities in the final packing of the algorithm. In the first case (case I), all critical bins receive a critical item, while in the second case (case II) some of them have their reserved space empty. 
In case I, we show the number of bins in the packing of \RRC is within a factor $1.5 + \frac{1-\beta}{4-3\beta}$ of the number of bins in the optimal packing.  
Note that this ratio decreases as the value of $\alpha$ (and $\beta)$ grows. This implies a less conservative algorithm would be better packing in this case.  Case II happens only if the advice is untrusted. In this case, the number of bins in the \RRC packing is within a factor $ 1.5 + \frac{9\beta}{8-6\beta}$ of the number of bins in an optimal packing. 
This ratio increases with $\alpha$ (and $\beta$). This implies a more conservative algorithm would be better in this case as it would open less critical bins and, thus, would have fewer without critical items. 

In what follows, we study the above cases in details to provide upper bounds for the competitive ratio of \RRC in the case of trusted and untrusted advice.


First, note that when $\gamma \leq \alpha$, 
then the algorithm 
works with the ratio $\gamma$ as indicated by the advice.
Consequently, if the advice is trusted, we have the same performance guarantee as stated in~\cite{ADKRR18}: 

\begin{lemma}\label{lem:pure-advice}
	\cite{ADKRR18} When $\gamma \leq \alpha$ and the advice is trusted, the competitive ratio of \textsc{Rrc} is at most $1.5 + \frac{15}{2^{k/2+1}}$.
\end{lemma}

The remaining cases are more interesting and involve scenarios when the advice is untrusted, or when  the advice is trusted but the algorithm maintains a ratio of $\alpha$ instead of $\gamma$ as indicated by the advice. 

First, we prove the following technical lemma that bounds the size of the packing of RRC based on the number of tiny bins and the parameter $\beta$. 


\begin{lemma}
	\label{lem:size-tiny}
	Let $S$ denote the total size of tiny items in an input sequence and assume there are $t$ tiny bins in the final packing of the \textsc{Rrc} algorithm. 
	We have $S > (t-1) \frac{4-3\beta}{6-6\beta}-1/6$.
\end{lemma}

\begin{proof}
	Assume $t>0$, otherwise the claim holds trivially. Recall that $\beta = \min\{\alpha,\gamma\}$, which gives $\beta < 1$ since $\alpha<1$.
	Let $B$ denote the last tiny bin that is opened by the algorithm and let $x$ be the tiny item which caused its opening. Let $c'$ and $t'$ respectively denote the number of critical and tiny bins before $B$ was opened ($t' = t-1$). Since $B$ is declared a tiny bin, we have $\frac{c'}{c'+t'} \geq \beta$ which gives $c'\geq \frac{\beta}{1-\beta} t'$.
	
	Since $x$ is tiny and caused the opening of a new bin, all $t'$ tiny bins have a level of at least $2/3$. Also we claim that all of the $c'$ critical bins, except possibly one 
	bin $B'$, contain tiny items of total size at least 1/6, 
	we call it the \emph{tiny level} of the bins. 
	If there are two critical bins with a tiny level of at most $1/6$, then each of them must contain at least one tiny item, otherwise $x$ could have fit.  And this means that one tiny item of the second bin could have fit into the first bin, contradicting the First-Fit packing of the algorithm.
	In summary, the total size $S$ of tiny items in the input sequence will be more than $t' \cdot 2/3$ (for tiny items in tiny bins) plus $(c'-1) \cdot 1/6$ 
	(for tiny items in critical bins). Since $c'> \frac{\beta}{1-\beta} t'$, we can write 
	$S > t' \cdot 2/3 + \beta/(1-\beta) t' \cdot 1/6 - 1/6 > t' (2/3 + \frac{\beta}{6(1-\beta)}) -1/6$.
\end{proof}

To continue our analysis of the \textsc{Rrc} algorithm, we investigate the two cases discussed earlier, captured by the following two lemmas. 
In this case, all bins declared as critical will receive a critical item.

\begin{lemma}
	\label{lem:bpl1}
	If all critical bins receive a critical item, then the number of bins in the final packing of \RRC algorithm is within a ratio 
	$1.5 + \frac{1-\beta}{4-3\beta}$ 
	of the number of bins in the optimal packing.
\end{lemma} 

\begin{proof}
	To prove the lemma, we use a weighting function argument as follows. Define the \emph{weight} of large and critical items to be 1, and the weight of small items to be 1/2.
	The weight of a tiny item of size $x$ is defined as $\frac{6-6\beta}{4-3\beta} x$. 
	Note that the weight of a tiny item $x$ is less than $3x/2$. 
	Let $W$ denote the total weight of all items in the sequence.
	
	First we claim that the number of bins opened by \textsc{Rrc} is at most $W+3$.
	Large bins include 1 large item of weight 1, and small 
	bins
	include two items of weight 1/2 (except possibly the last one) which gives a total weight of 1 for the bin. Critical bins all include a critical item of weight 1. So, if $w_\ell$, $w_s$, $w_c$ respectively denote the total weight of large, small, and critical items, then the number of non-tiny bins opened by the algorithm is at most $w_\ell+w_s+w_c+1$. Let $S$ denote the total size of tiny items. By Lemma~\ref{lem:size-tiny}, we have 
	$S > (t-1) \frac{4-3\beta}{6-6\beta} -1/6$. 
	The total weight of tiny bins is
	$\frac{6-6\beta}{4-3\beta} \cdot S \geq \frac{6-6\beta}{4-3\beta}  \cdot ((t-1) \frac{4-3\beta}{6-6\beta} -1/6) \geq t-2$.
	So, tiny items have total weight of at least $t-2$, that is, the number tiny bins is at most $w_t+2$, where $w_t$ is the total weight of tiny items. Consequently, the total number of bins opened by the algorithm is at most $w_\ell+w_s+w_c + 1 + w_t+2$, and the claim is established, i.e., \textsc{Rrc}$(\sigma) \leq W+3$.
	
	Next, we
	show the number of bins in an optimal solution is at least 
	$W (8-6\beta) / (14-11\beta)$. 
	For that, it suffices to show that the weight of any bin in the optimal solution (i.e., any collection of items with total size at most 1) is at most $(14-11\beta)/(8-6\beta)$. 
	Define the \emph{density} of an item as the ratio between its weight and size. To maximize the weight of a bin, it is desirable to place items of larger densities in it. This is achieved by placing a critical item of size $1/2+\epsilon$, a small item of size $1/3+\epsilon$ and a set of tiny items of total size $1/6-2\epsilon$ in the bin, where $\epsilon$ is an arbitrary small positive value. The weight of such a bin will be 
	$1 + 1/2 + (1/6-2\epsilon) \frac{6-6\beta}{4-3\beta} < \frac{14-11\beta}{8-6\beta} =  1.5 + \frac{1-\beta}{4-3\beta}$. 

	To summarize, we have \textsc{Rrc}$(\sigma) \leq W+3$ and 
	$\opt(\sigma) \geq \frac{8-6\beta}{14-11\beta}W$. 
	This gives an asymptotic competitive ratio of at most $\frac{14-11\beta}{8-6\beta} =  1.5 + \frac{1-\beta}{4-3\beta}$ 
	for \textsc{Rrc}.
\end{proof}


 Next, we consider case II, where the algorithm has declared too many bins as critical 
 and some of them did not receive any critical item.
\begin{lemma}
	\label{lem:bpl2}
	If some of the bins declared as critical do not receive a critical item, then the asymptotic number of bins in the final packing of \textsc{Rrc} algorithm within a ratio 
	$ 1.5 + \frac{9\beta}{8-6\beta}$ 
	of the number of bins in the optimal packing.
\end{lemma}

\begin{proof}
	Let $C$ denote the last critical bin opened by \textsc{Rrc}. Since there are critical bins without critical items at the final packing, $C$ should be opened by a tiny item. Let $x$ be the tiny item that opens $C$. Let $c'$ and $t'$ respectively denote the number of critical and tiny bins before $C$ is opened. Since $C$ is declared a critical bin, we have $\frac{c'}{c'+t'} < \beta$ which gives $c'< \frac{\beta}{1-\beta} t'$. As before, let $c$ and $t$ respectively denote the number of critical and tiny bins in the final packing ($c' = c-1$). We have $c < \frac{\beta}{1-\beta}t+1$.
	
	In order to prove the lemma, we show that all bins, except possibly a constant number of them, on average have a level of at least $\frac{4-3\beta}{6}$. This clearly holds for bins opened by large and small items, except possibly for the bin opened for the last small item; these bins all have a level of at least $2/3 \geq \frac{4-3\beta}{6}$. Note that if $c+t$ is a constant, 
	then all but a constant number of bins have level of at least 2/3 and the algorithm has a competitive ratio of at most 1.5. In what follows, we assume $c+t$ grows with the input length $n$. 
	By Lemma~\ref{lem:size-tiny}, the total size of tiny items is at least 
	$(t-1) \frac{4-3\beta}{6-6\beta}-1/6$. 
	These items are distributed between $t+c < t + \frac{\beta }{1-\beta}t+1 = (t-1)\frac{1}{1-\beta}+ d $ bins, for constant $d=(2-\beta)/(1-\beta)$. 
	Therefore, if we ignore $d$ bins, the average level of the remaining tiny/critical bins will be more than 
	$\frac{(4-3\beta)/(6-6\beta)}{1/(1-\beta)} - \frac{1}{6(t+c)} = \frac{4-3\beta}{6}- \frac{1}{6(t+c)}$. 
	Since $t+c$ is asymptotically large, the average level of these bins converges to a value of size larger than
	$\frac{4-3\beta}{6}$.
	
	Let $S$ denote the total size of items in the input sequence $\sigma$. Clearly at least $S$ bins are required to pack all items, i.e., $\opt(\sigma) \geq S$. On the other hand, since the average level of all bins (excluding $d+1$ bins) is more than 
	$\frac{4-3\beta}{6}$, 
	for the cost of \textsc{Rrc} we can write \textsc{Rrc}$(\sigma) \leq \lceil \frac{6}{4-3\beta} S \rceil + d+1 < \frac{6}{4-3\beta} S + d+2 \leq (1.5 + \frac{9\beta}{8-6\beta})\opt(\sigma) + d+2$.
\end{proof}

Provided with Lemmas~\ref{lem:bpl1} and \ref{lem:bpl2}, we are ready to prove an upper bound for the competitive ratio of \textsc{Rrc} as a function of its parameter $\alpha$. We consider two cases based on whether the advice is trusted.

\begin{lemma}\label{lem:bp-cor-adv}
	If the advice is trusted, then the competitive ratio of the \textsc{Rrc} algorithm is at most
	$1.5 + \max\{\frac{1-\alpha}{4-3\alpha}, \frac{15}{2^{k/2+1}} \}$.
			
\end{lemma}

\begin{proof}
	First, if $\gamma \leq \alpha$, by Lemma~\ref{lem:pure-advice}, the competitive ratio will be at most $1.5 +  \frac{15}{2^{k/2+1}}$. Next, assume $\alpha < \gamma$, that is $\beta = \alpha$. All critical bins receive a critical item in this case. This is because the algorithm maintains a critical ratio $\alpha$ which is smaller than $\gamma$. In other words, the algorithm declares a smaller ratio of its bins critical compared to the actual ratio in the Reserve-Critical algorithm. Hence, all critical bins receive a critical item. By 
	Lemma~\ref{lem:bpl1}, the competitive ratio is at most 
	$ 1.5+\frac{1-\alpha}{4-3\alpha}$.
\end{proof}

\begin{lemma}\label{lem:bp-incor-adv}
	If the advice is untrusted, then the competitive ratio of \textsc{Rrc} is at most
	$1.5+\max\{\frac{1}{4},\frac{9\alpha}{8-6\alpha}\}$.
\end{lemma}

\begin{proof}
	Consider two cases. First, assume all bins declared as critical by the \textsc{Rrc} algorithm receive a critical item. In this case, by Lemma~\ref{lem:bpl1}, the competitive ratio of the algorithm will be bounded by 
	$1.5 + \frac{1-\beta}{4-3\beta}$;
	 this value decreases in $\beta$ and hence is maximized at $\beta=0$. Next, assume some of the bins declared as critical do not receive a critical item. By Lemma~\ref{lem:bpl2}, the competitive ratio of \textsc{Rrc} in this case is at most 
	 $1.5 + \frac{9\beta}{8-6\beta}$; 
	 this value however increases by $\beta$ and is maximized at the upper bound $\beta=\alpha$.  This completes the proof.
\end{proof}

The following theorem 
directly follows from Lemmas~\ref{lem:bp-cor-adv} and \ref{lem:bp-incor-adv}.


%

\begin{theorem}\label{thm:bp-main}
Algorithm~Robust-Reserve-Critical with parameter $\alpha\in[0,1]$ and $k$ bits of advice achieves a competitive ratio of 
$r_\textsc{Rrc} \leq 1.5 + \max\{\frac{1-\alpha}{4-3\alpha}, \frac{15}{2^{k/2+1}} \}$ 
when the advice is trusted and a competitive ratio of $w_\textsc{Rrc} \leq 1.5+\max\{\frac{1}{4},\frac{9\alpha}{8-6\alpha}\}$ 
when the advice is untrusted.
\end{theorem}

Note that $r_\textsc{Rrc} \leq w_\textsc{Rrc}$, and therefore we have 
$r_\textsc{Rrc}\in [1.5,33/19]$. 
Assuming the size $k$ of the advice is a sufficiently large constant, we conclude the following. 
\begin{corollary}
For bin packing with untrusted advice, there is a $(r,f(r))$-competitive algorithm where $r \in [1.5,1.73]$ and 
$f(r) = \max \{33-18r, 7/4\}$.
\end{corollary}

\newcommand{\MTF}{\ensuremath{\textsc{Mtf}}\xspace}
\newcommand{\MTFF}{\ensuremath{\textsc{Mtf2}}\xspace}
\newcommand{\MTFE}{\ensuremath{\textsc{Mtfe}}\xspace}
\newcommand{\MTFO}{\ensuremath{\textsc{Mtfo}}\xspace}
\newcommand{\MO}{\ensuremath{\textsc{Tog}}\xspace}
\newcommand{\Mo}{\ensuremath{\textsc{Tog}}\xspace}
\newcommand{\OPT}{\ensuremath{\textsc{Opt}}\xspace}
\newcommand{\alg}{\ensuremath{\textsc{Alg}}\xspace}

\section{List update}
\label{sect:lu}

\subsection{Background}
\sh{In this section, we study the list update problem under the untrusted advice model.}The list update problem consists of a list of items of length $m$, and a sequence of $n$ requests that should be served with minimum total cost. Every request corresponds to an `access' to an item in the list. If the item is at position $i$ of the list then its access cost is $i$. After accessing the item, the algorithm can move it closer to the front of the list with no cost using a `free exchange'. In addition, at any point, the algorithm can swap the position of any two consecutive items in the list using a `paid exchange' which has a cost of 1. Throughout this section, we adopt the standard assumption that $m$ is a large integer but still a constant with respect to $n$. 
\sh{In particular, we assume that $m\in o(n)$.}

Move-to-Front (\MTF) is an algorithm that moves every accessed item to the front of the list using a free exchange.
{Sleator and Tarjan~\cite{SleTar85A} proved that }\MTF has a competitive ratio of at most 2~\cite{SleTar85A}, which is the best that a deterministic algorithm can achieve~\cite{Irani91}. \textsc{Timestamp}, introduced by Albers~\cite{Albers98}, is another algorithm that achieves the optimal competitive ratio of 2. This algorithm uses a free exchange to move an accessed item $x$ to the front of the first item that has been accessed at most once since the last access to $x$.
Another class of algorithms are Move-To-Front-Every-Other-Access (\MTFF) is a class of algorithms which maintain a bit for each item in the list. Upon accessing an item $x$, the bit of $x$ is flipped, and $x$ is moved to front if its bit is 0 after the flip (otherwise the list is not updated). If all bits are 0 at the beginning, 
\MTFF is called called Move-To-Front-Even (\MTFE), and 
if all bits are 1 at the beginning, \MTFF is 
called Move-To-Front-Odd (\MTFO). Both \MTFE and \MTFO algorithms have a competitive ratio of $5/2$~\cite{BoyKamLata14}. In \sh{Boyar et al}~\cite{BoyKamLata14} it is shown that, for any request sequence, at least one of \textsc{Timestamp}, \MTFO, and \MTFE has a competitive ratio of at most $5/3$.
For a given request sequence, the best option among the three algorithms can be indicated with two bits of advice, giving a $5/3$-competitive algorithm. However, if the advice is untrusted, the competitive ratio can be as bad as $5/2$.

To address this issue, we introduce an algorithm named Toggle (\MO) that has a parameter $\beta\in[0,1/2]$, and uses 2 advice bits to select one of the algorithms \textsc{Timestamp}, \MTFE or \MTFO, see Figure~\ref{fig:toggle}. This algorithm achieves a competitive ratio of $r_{\MO} = 5/3 + \frac{5\beta}{6+3\beta}$
when the advice is trusted and a competitive ratio of at most $w_{\MO} = 2 + 2/(4+5\beta)$ when the advice is untrusted. The parameter $\beta$ can be tuned and 
should be smaller when the advice is more reliable. In particular, when $\beta=0$, we get a $(5/3,5/2)$-competitive algorithm.

\begin{figure}[htb]
	\centering
\scalebox{.85}{
\begin{minipage}{\textwidth}
\begin{tikzpicture}
  \draw[->] (0, 0) -- (12.5, 0) node[right] {time};
  \def\sep{{0, 3, 5, 9, 12}};
  \foreach \i/\phase/\algo  in {1/trusting/Mtfo,2/ignoring/Mtf,3/trusting/Mtfo,4/ignoring/Mtf} {
	\draw (\sep[\i],-0.2) -- (\sep[\i],+0.2);
	\draw (\sep[\i-1],-0.2) -- (\sep[\i-1],+0.2);
	\draw ({\sep[\i]/2+\sep[\i-1]/2},0) node[below] {\phase}
		node[above] {\textsc{\algo}};
  };
\end{tikzpicture}
\\[1em]
\begin{tikzpicture}
  \draw[->] (0, 0) -- (12.5, 0) node[right] {time};
  \def\sep{{0, 2, 5.5, 8.5, 12.2}};
  \foreach \i/\phase/\algo  in {1/trusting/Mtfe,2/ignoring/Mtf,3/trusting/Mtfe,4/ignoring/Mtf} {
	\draw (\sep[\i],-0.2) -- (\sep[\i],+0.2);
	\draw (\sep[\i-1],-0.2) -- (\sep[\i-1],+0.2);
	\draw ({\sep[\i]/2+\sep[\i-1]/2},0) node[below] {\phase}
		node[above] {\textsc{\algo}};
  };
\end{tikzpicture}
\\[1em]
\begin{tikzpicture}
  \draw[->] (0, 0) -- (12.5, 0) node[right] {time};
  \draw (0,-0.2) -- (0,+0.2);
  \draw (6,0) node[above] {\textsc{Timestamp}};
\end{tikzpicture}
\end{minipage}}
\caption{The three different behaviors of the algorithm Toggle. The phases are indicated below the timeline.}
\label{fig:toggle}
\end{figure}



\subsection{The Toggle algorithm} 
Given the parameter $\beta$, the Toggle algorithm (\MO) works as follows. If the advice indicates \textsc{Timestamp}, the algorithm runs \textsc{Timestamp}. 
If the advice indicates either \MTFO or \MTFE, the algorithm will proceed in phases (the length of which partially depend on $\beta$) alternating (``toggling'') between running \MTFE or \MTFO, and \MTF. In what follows, we use \MTFF to represent the algorithm indicated by the advice{, that is, \MTFF is either \MTFE or \MTFO as per the advice}. The algorithm \MO will initially begin with \MTFF until the cost of the accesses of the phase reaches a certain threshold, then a new phase begins and \MO switches to \MTF. This new phase ends when the access cost of the phase reaches a certain threshold, and \MO switches back to \MTFF. This alternating pattern continues as \MO serves the requests. As such, \MO will use \MTFF for the odd phases which we will call \emph{trusting phases}, and \MTF for the even phases which we will call \emph{ignoring phases}. The actions during each phase are formally defined below.


\textit{Trusting phase: } In a trusting phase, \MO will use \MTFF to serve the requests. Let $\sigma_i$ be the first request of some trusting phase $j$ for $1 \le i \le n$ and an odd $j \ge 1$. Before serving
\sh{the first request of the phase}$\sigma_i$, \MO modifies the list with paid exchanges to match the list configuration that would result from running \MTFF on the request sequence $\sigma_1,\ldots,\sigma_i$. The number of paid exchanges will be less than $m^2$. In addition, \MO will set the bits of items in the list to the same value as at the end of this hypothetical run. As such, during a trusting phase, \MO incurs the same access cost as $\MTFF$. The trusting phase continues until the cost to access a request $\sigma_\ell$, $i < \ell \le n$, for \MO would cause the total access cost for the phase to become at least $m^3$ (or the request sequence ends). The next phase, which will be an ignoring phase, begins with request $\sigma_{\ell+1}$.

\textit{Ignoring phase: } In an ignoring phase, \MO will use the \MTF rule to serve the request.
In an ignoring phase, unlike the trusting phase, \MO does not use paid exchanges to match another list configuration. Let $\sigma_i$ be the first request of some ignoring phase $j$ for $1 \le i \le n$ and an even $j \ge 1$.
The ignoring phase continues until the cost to access a request $\sigma_\ell$, $i < \ell \le n$, for \MO would cause the total access cost for the phase to exceed $\beta \cdot m^3$ (or the request sequence ends). The next phase, which will be a trusting phase, begins with request $\sigma_{\ell+1}$.

%
%
%

\subsection{Analysis}

In our analysis, in the case of untrusted advice, we will focus on analyzing \MTFF. The reason for this is that, based on the competitive ratio, \textsc{Timestamp} has a competitive ratio of at most 2 which is better than the worst case of $2 + 2/(4+5\beta)$ that we will show when the untrusted advice indicates one of \MTFO or \MTFE.

Throughout the analysis, we fix a sequence $\sigma$ and use $k$ to denote the number of trusting phases of \Mo for serving $\sigma$. Note that the number of ignoring phases is either $k-1$ or $k$.
For each request, any algorithm incurs an access cost of at least 1 and hence each phase has length at most $m^3$. Since $m^3$ is a constant independent of the length of the input, $k$ grows with $n$. This observation will be used in the proof of the following two lemmas that help us bound the cost of \MO in the case of untrusted and trusted advice, respectively.

For the analysis, we will break the sequence into subsequences and analyze the cost over the subsequences. Let  $\sigma'$ be a subsequence of $\sigma$ and $\alg$ be any algorithm serving serving the sequence $\sigma$. We will denote the cost of $\alg$ over the subsequence $\sigma'$ with $\alg(\sigma')$, where it is implicit that $\alg$ has served the requests preceding $\sigma'$ in $\sigma$, and will serve the requests following $\sigma'$ in $\sigma$. The following lemma bounds the cost for an optimal algorithm over a subsequence as compared to a $c$-competitive online algorithm. 

\begin{lemma}\label{lem:phaseBoundOPT}
	Let $\alg$ be an online algorithm for the List Update problem such that, for all $\sigma$, $\alg(\sigma) \le c \cdot \opt(\sigma) + \alpha$. For any  $\sigma'$ that is a subsequence of $\sigma$,
	$$\OPT(\sigma') \ge \frac{\textsc{Alg}(\sigma')}{c} - \frac{\alpha}{c} - \frac{m^2}{2} + \frac{m}{2} ~.$$
\end{lemma}

\begin{proof}
	Let $r_i$ be the first request of $\sigma'$. Let $L^\alg_{r_j}$ be the list configuration of any algorithm $\alg$ immediately before serving the request $r_j$. Define $\opt'$ to be an optimal algorithm for the subsequence $\sigma'$ with an initial list configuration of $L^\alg_{r_i}$. That is, $\opt'$ is only serving $\sigma'$, starting from the configuration of $\alg$.
	
	Fix an optimal algorithm $\opt$ for $\sigma$. Define another algorithm
	$B$ that will only serve $\sigma'$. For a cost of at most $m(m-1)/2$ paid exchanges, $B$ will change its initial configuration of $L^\alg_{r_i}$ to $L^\opt_{r_i}$, serve $\sigma'$ as $\opt$ serves the subsequence in $\sigma$.
	The total cost of $B$ for $\sigma'$ cannot be less than $\opt'$ without contradicting the optimality of $\opt'$ for $\sigma'$. Hence, we have that, 
	\begin{align*}
		B(\sigma) &= \opt(\sigma') + m(m-1)/2 \ge \opt'(\sigma') ~.
	\end{align*}
	
	Using this and the competitive ratio of $\alg$, we get
	$$ \alg(\sigma') \le c \cdot \opt'(\sigma') + \alpha \le c \cdot (\opt(\sigma') + m(m-1)/2) + \alpha$$
	and the claim follows.
\end{proof}

\begin{lemma}
	\label{lem:costTrustMo}
	For a trusting phase, the cost of {\MO}  is in the range 
	$(m^3, m^3 (1+1/m+1/m^2))$ 
	(excluding the last phase). 
\end{lemma}

\begin{proof}
	For paid exchanges at the beginning of the phase, \MO incurs a cost that is less than $m^2$. Before serving the last request $\sigma_\ell$ of the phase, the access cost of \MO is less than $m^3$ by definition, and the access cost to $\sigma_\ell$ is at most $m$. 
\end{proof}

Similar arguments apply for an ignoring phase with the exception that the threshold is $\beta \cdot m^2$ and there are no paid exchanges performed by \MO. So, we can observe the following.

\begin{observation}\label{lem:costIgnoreMo}
	In an ignoring phase, the cost of \MO for the phase is in the range $(\beta m^3, \beta m^3(1+1/m^2))$	(excluding the last phase).
\end{observation}

The proof of the following lemma follows from Lemma~\ref{lem:costTrustMo} and Observation~\ref{lem:costIgnoreMo}, noting that there are $k$ trusting phases and at most $k$ ignoring phases.
\begin{lemma}\label{lem:costMo}
	The cost of \MO is bounded from above by $k \cdot m^3 \cdot \left(1 + \beta + \frac{3}{m} \right)$.
\end{lemma}

\begin{lemma}
	\label{lem:crMo}
	For sufficiently long lists and long request sequences, the competitive ratio of \MO (regardless of the advice being trusted or not) converges to at most $2 + \frac{2}{4+5\beta}$. 
\end{lemma}

\begin{proof}
	Consider an arbitrary trusting phase and let $\sigma_t$ denote the subsequence of $\sigma$ formed by requests in that phase. Recall that \MO uses the \MTFF strategy during a trusting phase. 
	We know that $\MTFF(\sigma_t) \leq 2.5\cdot \OPT(\sigma_t) + m^2$~\cite{BoyKamLata14}, and that \MTFF incurs a cost of more than $m^3$ during the phase (Lemma~\ref{lem:costTrustMo}). So, from Lemma~\ref{lem:phaseBoundOPT}, we conclude that \OPT incurs a cost of at least $m^3/2.5 - m^2$ during the phase. Note that this lower bound for the cost of \OPT applies for all trusting phases.
	
	Next, consider an arbitrary ignoring phase and let $\sigma'$ denote the subsequence of requests served by \MO during that phase. Recall that \MO applies \MTF during an ignoring phase.
	We know $\MTF(\sigma') \leq 2 \textsc{\OPT}(\sigma')+m^2$~
	\cite{SleTar85}, and \MTF incurs an access cost of at least $\beta m^3$ during the phase (Lemma~\ref{lem:costIgnoreMo}). So, from Lemma~\ref{lem:phaseBoundOPT}, we conclude that \OPT incurs a cost of at least $\beta m^3/2 - m^2$ during the phase. Note that this lower bound for the cost of \OPT applies for all ignoring phases.
	
	Since we have at least $k-1$ of each trusting and ignoring phases, the total cost of \OPT is at least $(k-1) (m^3/2.5 - m^2) + (k-1) (\beta m^3/2 - m^2) = (k-1) (m^3 \frac{4 + 5\beta}{10} -2m^2) > (k-1)m^3 (\frac{4 + 5\beta}{10} -2/m)$.
	
	To summarize, the cost of \OPT is larger than $(k-1)m^3 (\frac{4 + 5\beta}{10} -2/m)$ and, by Lemma~\ref{lem:costMo}, the cost of \MO is at most $k m^3(1+\beta + 3/m)$. The competitive ratio will be at most $\frac{k m^3(1+\beta + 3/m)}{(k-1)m^3 (\frac{4 + 5\beta}{10} -2/m)}$ which converges to $\frac{10+10\beta}{4+5\beta} = 2+2/(4+5\beta)$ for long lists (as $m$ grows) and long input sequences (as $k$ grows).
\end{proof}

\begin{lemma}
	\label{lem:luCorAd}
	For sufficiently long lists, the ratio between the cost of \MO and that of \MTFF converges to $1+\frac{\beta}{2+\beta}$.
\end{lemma}

\begin{proof}
	Note that \MO and \MTFF incur the same access cost of $m^3$ in any trusting phases. We use an argument similar to the previous lemma for analyzing ignoring phases. Consider an arbitrary ignoring phase and let $\sigma'$ denote the subsequence of requests served by \MO during that phase.
	We know $\MTF(\sigma') \leq 2\cdot \OPT(\sigma') + m^2$~\cite{BoyKamLata14}, and
	\MTF incurs a cost of at least $\beta m^3$ during the phase. So, from Lemma~\ref{lem:phaseBoundOPT}, we conclude that \OPT, and consequently \MTFF, incur a cost of at least $\beta m^3/2 - m^2$ during the phase. Note that this lower bound for the cost of \MTFF applies for all ignoring phases.
	
	The worst-case ratio between the costs of \MO and \MTFF is maximized when the last phase is an ignoring phase. In this case, we have $k$ trusting phases and $k$ ignoring phases. The total cost of \MTFF is at least $k m^3 + k (\beta m^3/2 -m^2) = km^3 ( 1+\beta/2 -1/m)$. By Lemma~\ref{lem:costMo}, the cost of \MO is at most $k m^3(1+\beta + 3/m)$. The ratio between the two algorithms will be less than
	$\frac{k m^3(1+\beta + 3/m)}{km^3 ( 1+\beta/2 -1/m)}$ which converges to $1+\frac{\beta}{2+\beta}$ for long lists.
\end{proof}

Given the above lemmas, we can find upper bounds on the competitive ratio of \MO.
\begin{lemma}
	\label{lem:luCor}
	If the advice is trusted, then the competitive ratio of the \MO algorithm converges to $5/3 + \frac{5\beta}{6+3\beta}$ for sufficiently long lists.
\end{lemma}

\begin{proof}
	If the advice indicates \textsc{Timestamp} as the best algorithm among \MTFE, \MTFO, and \textsc{Timestamp}, the algorithm uses \textsc{Timestamp} to serve the entire sequence, and since the advice is right, the competitive ratio will be at most 5/3~\cite{BoyKamLata14}. If the advice indicates \MTFE or \MTFO as the best algorithm, the \MO algorithm uses the phasing scheme described above by alternating between the indicated algorithm and \MTF. If the advice is right, by Lemma~\ref{lem:luCorAd}, the cost of the algorithm will be within a ratio $1+\frac{\beta}{2+\beta}$ of the algorithm indicated by the advice, and consequently has a competitive ratio of at most $5/3 (1+\frac{\beta}{2+\beta}) = 5/3 + \frac{5\beta}{6+3\beta}$.
\end{proof}

\begin{lemma}
	If the advice is untrusted, then the competitive ratio of the \MO algorithm converges to $2 + \frac{2}{4+5\beta}$ for sufficiently long lists.
	\label{lem:luIncor}
\end{lemma}

\begin{proof}
	If the advice indicates \textsc{Timestamp} as the best algorithm, the algorithm trusts it and the competitive ratio will be at most 2~\cite{Albers98}. If the advice indicates \MTFE or \MTFO as the best algorithm, the \MO algorithm uses the phasing scheme described above by alternating between the indicated algorithm and Move-To-Front, and by Lemma~\ref{lem:crMo}, the competitive ratio of the algorithm will be at most $2 + \frac{2}{4+5\beta}$.
\end{proof}

The following theorem directly follows from Lemmas~\ref{lem:luCor} and \ref{lem:luIncor}.

\begin{theorem}\label{th:lu-main}
Algorithm \MO with parameter $\beta\in[0,1/2]$ and $2$ bits of advice achieves a competitive ratio of at most $5/3 + \frac{5\beta}{6+3\beta}$ when the advice is trusted and a competitive ratio of at most $2 + \frac{2}{4+5\beta}$ when the advice is untrusted.
\end{theorem}

\begin{corollary}\label{coro:luDet}
For list update with untrusted advice, there is a $(r,f(r))$-competitive algorithm where $r\in[5/3,2]$ and $f(r) = 2+ \frac{10 - 3r}{9r - 5}$.
\end{corollary}

\section{Randomized online algorithms with untrusted advice}
\label{sec:extensions}

The discussion in all previous sections pertains to deterministic online algorithms. In this section,we focus on randomization and its
impact on online computation with untrusted advice. We will assume, as standard in the analysis of randomized algorithms, that the source
of randomness is trusted (unlike the advice). Given a randomized algorithm $A$, its trusted and untrusted competitive ratios 
are defined as in~\eqref{eq:definition}, with the difference that the cost $A(\sigma,\phi)$ is now replaced by the expected cost $\E(A(\sigma,\phi))$.

Randomization can improve the competitiveness of the ski rental problem~\cite{NIPS2018_8174,wei2020optimalold}.
Namely,~\cite{NIPS2018_8174} gave a randomized algorithm with a single advice bit which is 
$\big(\frac{\lambda}{1-e^{-\lambda}}, \frac{1}{1-e^{-(\lambda-1/B)}}\big)$-competitive, where $\lambda\in (1/B,1)$ is a parameter of the algorithm. More specifically, the algorithm uses as advice a bit that predicts whether the length of the sequence (i.e., the number of skiing days) exceeds $B$ or not, and samples the day when skis are bought based on two different probability distributions,
depending on the advice bit. For simplicity, we may assume that $B$ is large, hence this algorithm is $\big(\frac{\lambda}{1-e^{-\lambda}}, \frac{1}{1-e^{-\lambda}}\big)$-competitive,
which we can write in the equivalent form $(w\ln\frac{w}{w-1},w)$.
In contrast, Theorem~\ref{th:ski-main} shows that any deterministic Pareto-optimal algorithm with advice of any size is 
$(1+\lambda,1+1/\lambda)$-competitive, or equivalently $(\frac{w}{w-1},w)$-competitive.  Since $ w\ln\frac{w}{w-1}<\frac{w}{w-1}$, we conclude that the randomized algorithm Pareto-dominates any deterministic algorithm, even when the latter is allowed unbounded advice.



A second issue we address in this section is related to the comparison of random bits and advice bits as resource. More specifically, 
in the standard model in which advice is always trustworthy, an advice bit can be at least as powerful as a random bit since
the former can simulate the efficient choice of the latter, and thus provide a ``no-loss'' derandomization. However, in the setting of untrusted advice, 
the interplay between advice and randomization is much more intricate. This is because random bits, unlike advice bits, are assumed to be trusted.

We show, using online bidding as an example, that there are situations in which a deterministic algorithm with $L+1$ advice bits is Pareto-incomparable to a
randomized algorithm with 1 random bit and $L$ advice bits. In particular we focus on the {\em bounded online bidding} problem, 
in which $u\leq B$, for some given $B$.

\begin{theorem}
For every $\epsilon>0$ there exist sufficiently large $B$ and $L$ such that there is a randomized algorithm for bounded online bidding 
with $L$ advice bits and 1 random bit that is $(\frac{1+\rho_1}{2} \rho_1+\epsilon , \alpha w+\epsilon)$-competitive 
for all $w>4$, where $\rho_1=\frac{w-\sqrt{w^2-4w}}{2}$, and 
$\alpha=\max\{ \frac{1+\rho_1}{2}, \frac{\rho_1+\rho_2}{2\rho_2}\}$.
\label{thm:randomized.bidding}
\end{theorem}

\begin{proof}
	Recall that the strategy $X_u^*=(x_i)$ of Theorem~\ref{thm:bidding.optimal} is $(\frac{w-\sqrt{w^2-4w}}{2},w)$-competitive, and that
	for $w>4$, its bids are as in the statement of Lemma~\ref{lem:bidding.recurrence}. 
	To simplify the analysis, we will rely on a different characterization due to~\cite{DBLP:journals/jair/AngelopoulosK23}, which showed that $X_u^*$ can be expressed as a scaled geometric strategy with base $\rho_2$, namely, we have that $x_i=\lambda \rho_2^i$, for some $\lambda$ that only depends on the target $u$. Define strategy $Y=(y_i)_{i \geq 1}$ in which $y_i=\rho_1 x_i$. Last, we define a randomized strategy $R$ with 1 random bit, which mixes equiprobably between $X$ and $Y$. 
	
	There is a subtlety concerning the advice bits required to encode $u$. 
	First, we note that for every $\epsilon>0$, there is sufficiently large $B$ with the following property:
	if a deterministic algorithm has competitive ratio $w$ for unbounded $u$, then it has competitive ratio at 
	least $w-\epsilon$ when $u$ is constrained to be at most $B$. Furthermore, if $u$ is bounded, then with a sufficiently
	large number of bits, say $b_u$, we can approximate $u$ to any required precision. This in turn implies that if $u$
	is sufficiently large, then an encoding with $b_u$ and an encoding with $b_u-1$ bits can differ only by $\epsilon$,
	and therefore so do the competitive ratios of the algorithms $X^*_{u_1}$ and $X^*_{u_2}$, where $u_1$ and $u_2$ are the
	values encoded with $b_u$ and $b_u-1$ bits, respectively. Summarizing, we can assume that, excluding negligible effects to the trusted and untrusted
	competitive ratios, $X_u^*$ and $Y$ receive the precise value of $u$ as advice, and that $X_{u}^*$ requires one bit more than $Y$.
	
	Note that, by definition,  $y_i > x_i$, and that $y_i \leq x_{i+1}$.
	The worst case choices for the hidden value $u$ are values infinitesimally larger than the ones in the sets
	$\cup_i\{x_i\}$ and $\cup_i \{y_i\}$. We thus consider two cases:
	
	\medskip
	\noindent
	{\em Case 1:} $u$ is infinitesimally larger than $x_i$. Since $x_{i+1} >u$, and $y_{i}=\rho_1 x_i >u$, the expected cost of $R$
	is at most 
	\[
	\frac{1}{2} \sum_{j=1}^{i+1}x_j+ \frac{1}{2} \sum_{j=1}^{i}y_j.
	\] 
	
	\medskip
	\noindent
	{\em Case 2:} $u$ is infinitesimally larger than $y_i$. Since $y_{i+1} >u$ and $x_{i+1}>y_i$, 
	Thus the expected cost of $R$ is at most
	\[
	\frac{1}{2} \sum_{j=1}^{i+1} (x_j+y_j)=\frac{1+\rho_1}{2} \sum_{j=1}^{i+1} x_j.
	\]

	We can now bound the expected untrusted competitive ratio of $R$. If $u$ is from the set $\cup_i\{x_i\}$, then 
	\begin{align*}
		\frac{\E(cost(R))}{\opt} &=\frac{1}{2} \cdot \sup_i \frac{\sum_{j=1}^{i+1}x_j+\sum_{j=1}^{i}y_j}{x_i} \nonumber  \\
		&\leq 
		\frac{1}{2} \left( \sup_i \frac{\sum_{j=1}^{i+1} x_j}{x_{i}}  + \rho_1 \sup_i \frac{\sum_{j=1}^i y_{j}}{y_i} \right)
		\nonumber \tag{$y_i=\rho_1 x_i$}\\
		&= \frac{1}{2} \left( w  +\rho_1 \frac{\rho_2}{\rho_2-1} \right) \nonumber \tag{From $y_i=\lambda \rho_1 \rho_2^i$} \\
		&=\frac{1}{2}(w+\frac{\rho_1}{\rho_2}w)=w \frac{\rho_1+\rho_2}{2\rho_2}.
	\end{align*}
	
	If $u$ is from the set $\cup_i\{y_i\}$, then 
	\begin{align*}
		\frac{\E(cost(R))}{\opt} &= \frac{1+\rho_1}{2} \cdot \sup_i \frac{\sum_{j=1}^{i+1} x_j}{y_i} \nonumber \\
		&= \frac{1+\rho_1}{2} \cdot \sup_i \frac{\sum_{j=1}^{i+1} x_j}{\rho_1 x_i} =\frac{1+\rho_1}{2\rho_1} w.
	\end{align*}
	
	Summarizing, we have 
	\[
	w_R \leq \max \{w \frac{\rho_1+\rho_2}{2\rho_2}, w \frac{1+\rho_1}{2\rho_1} \}= \alpha w,
	\]
where recall that $\alpha=\max\{ \frac{1+\rho_1}{2}, \frac{\rho_1+\rho_2}{2\rho_2}\}$. Note also that for any $w>4$, $\alpha<1$.

	Concerning the expected trusted competitive ratio of $R$, we observe that with probability $1/2$ it is equal to $\rho_1$,  
	(if $X^*_u$ is chosen), or equal to at most $\rho_1 \rho_1=\rho_1^2$, if $Y$ is chosen.  Thus, 
	\[
	r_R\leq \rho_1\frac{1+\rho_1}{2}.
	\]
\end{proof}

Note that when $B, L \rightarrow \infty$, the competitiveness of the best deterministic algorithm with $L$ advice bits approaches the one of $X_u^*$, as
expressed in Theorem~\ref{thm:bidding.optimal}, namely $(\rho_1,w)$. Thus, 
Theorem~\ref{thm:randomized.bidding} shows that randomization improves upon the deterministic untrusted ratio $w$ by a multiplicative factor $\alpha$, at the expense of a degradation of the trusted competitive ratio by a factor $\frac{1+\rho_1}{2}>1$. 


\section{Conclusion}

\label{sec:conclusion}

We introduced a new model in the study of online algorithms with advice, in which the online algorithm can leverage information about the request sequence that is not necessarily foolproof. Motivated by advances in learning-online algorithms, we studied tradeoffs between the trusted and untrusted competitive ratio, as function of the advice size. We also proved the first lower bounds for online algorithms in this setting. Any other online problem should be amenable to analysis under this framework, and in particular any other of the many problems studied under the classic framework of (standard) advice complexity. 

In future work, we would like to expand the model so as to incorporate, into the analysis, the concept of advice {\em error}. More specifically, given an advice string of size $k$, let $\eta$ denote the number of erroneous bits (which may be not known to the algorithm). In this setting, the objective would be  to study the power and limitations of online algorithms, i.e., from the point of view of both upper and lower bounds on the competitive ratio. A first approach towards this direction was made recently in the context of problems such as contract 
scheduling~\cite{DBLP:journals/jair/AngelopoulosK23} and online time-series search~\cite{DBLP:conf/aaai/0001KZ22}. Considering that each advice bit can be interpreted as a response to a {\em binary query}, such studies may provide connections with areas such as computation with {\em noisy queries} (see., e.g.,~\cite{NIPS2017_7161} for clustering in a setting in which a stochastic oracle answers whether two input points should be part of the same cluster). This opens up the opportunity for bringing the advice complexity model much closer to real-world applications.

\bibliographystyle{abbrv} 
\bibliography{online}

\begin{thebibliography}{10}

\bibitem{Albers98}
S.~Albers.
\newblock Improved randomized on-line algorithms for the list update problem.
\newblock {\em SIAM J. Comput.}, 27:682--693, 1998.

\bibitem{ADKRR18}
S.~Angelopoulos, C.~D{\"{u}}rr, S.~Kamali, M.~P. Renault, and A.~Ros{\'{e}}n.
\newblock Online bin packing with advice of small size.
\newblock {\em Theory Comput. Syst.}, 62(8):2006--2034, 2018.

\bibitem{DBLP:journals/jair/AngelopoulosK23}
S.~Angelopoulos and S.~Kamali.
\newblock Contract scheduling with predictions.
\newblock {\em J. Artif. Intell. Res.}, 77:395--426, 2023.

\bibitem{DBLP:journals/jair/AngelopoulosKS23}
S.~Angelopoulos, S.~Kamali, and K.~Shadkami.
\newblock Online bin packing with predictions.
\newblock {\em J. Artif. Intell. Res.}, 78:1111--1141, 2023.

\bibitem{DBLP:conf/aaai/0001KZ22}
S.~Angelopoulos, S.~Kamali, and D.~Zhang.
\newblock Online search with best-price and query-based predictions.
\newblock In {\em Proceedings of the 36th AAAI Conference on Artificial
  Intelligence}, pages 9652--9660. {AAAI} Press, 2022.

\bibitem{BaloghBDEL18}
J.~Balogh, J.~B{\'{e}}k{\'{e}}si, G.~D{\'{o}}sa, L.~Epstein, and A.~Levin.
\newblock A new and improved algorithm for online bin packing.
\newblock In {\em Proceedings of the 26th Annual European Symposium on
  Algorithms (ESA)}, pages 5:1--5:14, 2018.

\bibitem{BaloghArX05554}
J.~Balogh, J.~B{\'{e}}k{\'{e}}si, G.~D{\'{o}}sa, L.~Epstein, and A.~Levin.
\newblock A new lower bound for classic online bin packing.
\newblock {\em Algorithmica}, 83(7):2047--2062, 2021.

\bibitem{BocKomKra09}
H.-J. B\"{o}ckenhauer, D.~Komm, R.~Kr{\'a}lovi{\v{c}}, R.~Kr{\'{a}}lovi{\v{c}},
  and T.~M{\"{o}}mke.
\newblock On the advice complexity of online problems.
\newblock In {\em Proceedings of the 20th International Symposium on Algorithms
  and Computation (ISAAC)}, pages 331--340, 2009.

\bibitem{Boyar:survey:2016}
J.~Boyar, L.~M. Favrholdt, C.~Kudahl, K.~S. Larsen, and J.~W. Mikkelsen.
\newblock Online algorithms with advice: A survey.
\newblock {\em SIGACT News}, 47(3):93--129, 2016.

\bibitem{BoyarFavrholdt:15:Advice}
J.~Boyar, L.~M. Favrholdt, C.~Kudahl, and J.~W. Mikkelsen.
\newblock Advice complexity for a class of online problems.
\newblock In {\em Proceedings of the 32nd International Symposium on
  Theoretical Aspects of Computer Science ({STACS})}, 2015.

\bibitem{BoyKamLata14}
J.~Boyar, S.~Kamali, K.~S. Larsen, and A.~L{\'o}pez-Ortiz.
\newblock On the list update problem with advice.
\newblock {\em Information and Computation}, 253:411--423, 2016.

\bibitem{BoyarKLL16}
J.~Boyar, S.~Kamali, K.~S. Larsen, and A.~L{\'{o}}pez{-}Ortiz.
\newblock Online bin packing with advice.
\newblock {\em Algorithmica}, 74(1):507--527, 2016.

\bibitem{ChrKen06}
M.~Chrobak and C.~Kenyon.
\newblock Competitiveness via doubling.
\newblock {\em SIGACT News}, pages 115--126, 2006.

\bibitem{BPsurvey2013}
E.~G. Coffman~Jr., J.~Csirik, G.~Galambos, S.~Martello, and D.~Vigo.
\newblock Bin packing approximation algorithms: survey and classification.
\newblock In P.~M. Pardalos, D.-Z. Du, and R.~L. Graham, editors, {\em Handbook
  of Combinatorial Optimization}, pages 455--531. Springer, 2013.

\bibitem{DobKraPar09}
S.~Dobrev, R.~Kr\'{a}lovi\v{c}, and D.~Pardubsk\'{a}.
\newblock Measuring the problem-relevant information in input.
\newblock {\em RAIRO Inform. Theor. Appl.}, 43(3):585--613, 2009.

\bibitem{EmekFraKorRos2011}
Y.~Emek, P.~Fraigniaud, A.~Korman, and A.~Ros\'{e}n.
\newblock Online computation with advice.
\newblock {\em Theoret. Comput. Sci.}, 412(24):2642 -- 2656, 2011.

\bibitem{Irani91}
S.~Irani.
\newblock Two results on the list update problem.
\newblock {\em Inform. Process. Lett.}, 38:301--306, 1991.

\bibitem{Johnso73}
D.~S. Johnson.
\newblock {\em Near-optimal bin packing algorithms}.
\newblock PhD thesis, MIT, Cambridge, MA, 1973.

\bibitem{Johnson16}
D.~S. Johnson.
\newblock Bin packing.
\newblock In {\em Encyclopedia of Algorithms}, pages 207--211. 2016.

\bibitem{KamLopDCC14}
S.~Kamali and A.~L\'opez-Ortiz.
\newblock Better compression through better list update algorithms.
\newblock In {\em Data Compression Conference ({DCC})}, pages 372--381, 2014.

\bibitem{KaMaRS88}
A.~Karlin, M.~Manasse, L.~Rudolph, and D.~Sleator.
\newblock Competitive snoopy caching.
\newblock {\em Algorithmica}, 3:79--119, 1988.

\bibitem{karlin2003dynamic}
A.~R. Karlin, C.~Kenyon, and D.~Randall.
\newblock Dynamic {TCP} acknowledgment and other stories about e/(e - 1).
\newblock {\em Algorithmica}, 36:209--224, 2003.

\bibitem{aaai06:contracts}
A.~L\'opez-Ortiz, S.~Angelopoulos, and A.~M. Hamel.
\newblock Optimal scheduling of contract algorithms for anytime problems.
\newblock {\em Journal of Artificial Intelligence Research}, 51:533--554, 2014.

\bibitem{DBLP:conf/icml/LykourisV18}
T.~Lykouris and S.~Vassilvitskii.
\newblock Competitive caching with machine learned advice.
\newblock {\em J. {ACM}}, 68(4):24:1--24:25, 2021.

\bibitem{NIPS2017_7161}
A.~Mazumdar and B.~Saha.
\newblock Clustering with noisy queries.
\newblock In {\em Annual Conference on Neural Information Processing Systems
  {(NIPS)}}, volume~30, pages 5788--5799. 2017.

\bibitem{Meyers05}
A.~Meyerson.
\newblock The parking permit problem.
\newblock In {\em Proceedings of the 46th Annual IEEE Symposium on Foundations
  of Computer Science ({FOCS})}, pages 274--284, 2005.

\bibitem{DBLP:conf/icalp/Mikkelsen16}
J.~Mikkelsen.
\newblock Randomization can be as helpful as a glimpse of the future in online
  computation.
\newblock In {\em Proceedings of the 43rd International Colloquium on Automata,
  Languages, and Programming {(ICALP)}}, pages 39:1--39:14, 2016.

\bibitem{DBLP:books/cu/20/MitzenmacherV20}
M.~Mitzenmacher and S.~Vassilvitskii.
\newblock Algorithms with predictions.
\newblock In {\em Beyond the Worst-Case Analysis of Algorithms}, pages
  646--662. Cambridge University Press, 2020.

\bibitem{NIPS2018_8174}
M.~Purohit, Z.~Svitkina, and R.~Kumar.
\newblock Improving online algorithms via {ML} predictions.
\newblock In {\em Proceedings of the 32nd Conference on Advances in Neural
  Information Processing Systems ({NeurIPS})}, volume~31, pages 9661--9670,
  2018.

\bibitem{RenaultRS15}
M.~P. Renault, A.~Ros{\'{e}}n, and R.~van Stee.
\newblock Online algorithms with advice for bin packing and scheduling
  problems.
\newblock {\em Theor. Comput. Sci.}, 600:155--170, 2015.

\bibitem{RZ.1991.composing}
S.~J. Russell and S.~Zilberstein.
\newblock Composing real-time systems.
\newblock In {\em Proceedings of the 12th International Joint Conference on
  Artificial Intelligence (IJCAI)}, pages 212--217, 1991.

\bibitem{SleTar85}
D.~Sleator and R.~E. Tarjan.
\newblock Amortized efficiency of list update and paging rules.
\newblock {\em Commun. ACM}, 28:202--208, 1985.

\bibitem{SleTar85A}
D.~D. Sleator and R.~E. Tarjan.
\newblock Self-adjusting binary search trees.
\newblock {\em Jour. of the ACM}, 32:652--686, 1985.

\bibitem{wei2020optimalold}
A.~Wei and F.~Zhang.
\newblock Optimal robustness-consistency trade-offs for learning-augmented
  online algorithms.
\newblock {\em arXiv preprint arXiv:2010.11443}, 2020.

\end{thebibliography}



\vspace{1cm}

\appendix


\section{Omitted proofs from Section~\ref{sec:online.bidding}}


Given $u, m \geq 1$, assuming that $(L_{m, u})$ is feasible, we will first show how to compute the optimal objective value of $(L_{m,u})$. Let \text{Obj} denote the numerator of objective value of $(L_{m,u})$ (namely, $\sum_{j=1}^i x_j$).
For convenience, we denote $x_{m, u, i}$ by $x_{i}$ and let $T_i$ denote $\sum_{j=1}^i x_j$, with $T_0 = 0$.

Define the sequences $c_i$ and $d_i$ as follows:
\begin{equation}
  \label{eq:rec_c}
  c_i = c_{i-1} + d_{i-1} \cdot a_{i-1}, \textrm{ with } c_0 = 0,
\end{equation}

\begin{equation}
  \label{eq:rec_d}
  d_i = d_{i-1} \cdot (1+ b_{i-1}), \textrm{ with } d_0 = 1.
\end{equation}
Sequences $a_i$ and $b_i$ were defined in \eqref{eq:rec_a'} and \eqref{eq:rec_b} right before Lemma~\ref{lem:bidding.recurrence}.

The sequences $a_i, b_i, c_i$ and $d_i$ satisfy the following technical properties.
\begin{lemma}
  \label{lem:a_b_c_d}
  For $i \geq 0$, we have
  \[
  a_i = \left\{
  \begin{array}{lr}
    \frac{2}{i+2} \cdot \frac{1}{2^i}, & w = 4  \\
    \frac{p^2-1}{p^{i+2}-1} \cdot (\frac{p}{w})^{\frac{i}2}, & w > 4,
  \end{array}
  \right. \quad  b_i = \left\{
  \begin{array}{lr}
    \frac{i}{i+2}, & w = 4  \\
    p\cdot \frac{p^i-1}{p^{i+2}-1}, & w > 4,
  \end{array}
  \right.
  \]
  
  \[
  c_i = \left\{
  \begin{array}{lr}
    2 - \frac{2}{i+1}, & w = 4  \\
    1 + p - \frac{p^i(p^2-1)}{p^{i+1}-1}, & w > 4
  \end{array}
  \right. \quad \textrm{ and } \quad d_i = \left\{
  \begin{array}{lr}
    \frac{2^{i}}{i+1} , & w = 4  \\
    \frac{p-1}{p^{i+1}-1} \cdot (pw)^{\frac{i}2}, & w > 4,
  \end{array}
  \right.
  \]
  with $p = \frac{w-2-\sqrt{w^2-4w}}{2}$.
\end{lemma}

\begin{proof}
  Choose $p < 1$ such that
  \begin{equation}
    \label{eq:p}
    p = \frac{1+p}{w-1-p}.
  \end{equation}
  In other words, $p = \frac{w-2-\sqrt{w^2-4w}}{2}$. From~\eqref{eq:rec_b}, we have
  \[
  b_i - p = \frac{(p+1)(b_{i-1} - p)}{w-1-p - (b_{i-1} - p)},
  \]
  which implies that
  \[
  \frac{1}{b_i - p} = \frac{w-1-p}{p+1}\cdot \frac{1}{b_{i-1} - p} - \frac{1}{p+1}.
  \]
  Define the sequence $(u_i)_{i \geq 0}$ as $u_i = \frac{1}{b_i - p}$ for $i \geq 0$, then
  \[
  u_i = \frac{1}{p} \cdot u_{i-1} - \frac{1}{p+1}, \quad \textrm{ with } \quad u_0 = \frac{-1}{p}.
  \]
  Thus,
  \[
  u_i = \left\{
  \begin{array}{lr}
    -\frac{i+2}{2}, & w = 4  \\
    -\frac{p^{i+2}-1}{(p^2-1)p^{i+1}}, & w > 4,
  \end{array}
  \right.
  \]
  which implies that
  \[
  b_i = \left\{
  \begin{array}{lr}
    \frac{i}{i+2}, & w = 4  \\
    p\cdot \frac{p^i-1}{p^{i+2}-1}, & w > 4.
  \end{array}
  \right.
  \]
  Then
  \begin{equation}
    \label{eq:prod_bi}
    \prod_{j=1}^{i} b_j = \left\{
    \begin{array}{lr}
      \frac{2}{(i+1)(i+2)}, & w = 4  \\
      p^i\cdot \frac{(p-1)(p^2-1)}{(p^{i+1}-1)(p^{i+2}-1)}, & w > 4,
    \end{array}
    \right.
  \end{equation}
  In addition, from~\eqref{eq:rec_a'} and~\eqref{eq:rec_d}, for $i \geq 1$, we have
  \[
  a_i = \prod_{j=1}^i \frac{1}{w-1-b_{j-1}} \quad \textrm{ and } \quad d_i = \prod_{j=1}^i (1+b_{j-1}).
  \]
  Then, for $i \geq 2$,
  \begin{equation}
    \label{eq:a_id_i}
    a_{i}d_{i} = \prod_{j=1}^{i} \frac{(1+b_{j-1})}{w-1-b_{j-1}} = \prod_{j=1}^{i} b_j.
  \end{equation}
  Moreover, from~\eqref{eq:rec_b}, we have
  \[
  1 + b_i = \frac{w}{w-1-b_{i-1}},
  \]
  then
  \[
  \prod_{j=1}^i (1+b_{j}) = w^i \cdot \prod_{j=1}^i \frac{1}{w-1-b_{j-1}},
  \]
  which implies that
  \begin{equation}
    \label{eq:d_i_a_i}
    d_{i+1} = w^i \cdot a_i.
  \end{equation}
  Combining~\eqref{eq:rec_d},~\eqref{eq:a_id_i} and~\eqref{eq:d_i_a_i}, we have
  \[
  a_i = \sqrt{\frac{(1+b_{i}) \cdot \prod_{j=1}^i b_j}{w^i} } \quad \textrm{ and } \quad d_i = \sqrt{\frac{w^i \cdot \prod_{j=1}^i b_j}{1+b_i} }.
  \]
  Thus, if $w > 4$,
  \[
  a_i = \sqrt{\frac{(1+p\cdot \frac{p^i-1}{p^{i+2}-1}) \cdot p^i\cdot \frac{(p-1)(p^2-1)}{(p^{i+1}-1)(p^{i+2}-1)} }{w^i} } = \frac{p^2-1}{p^{i+2}-1} \cdot \left(\frac{p}{w}\right)^{\frac{i}2}
  \]
  and
  \[
  d_i = \sqrt{ \frac{w^i \cdot p^i\cdot \frac{(p-1)(p^2-1)}{(p^{i+1}-1)(p^{i+2}-1)}}{(1+p\cdot \frac{p^i-1}{p^{i+2}-1})}} = \frac{p-1}{p^{i+1}-1} \cdot (pw)^{\frac{i}2}.
  \]
  If $w  = 4$,
  \[
  a_i = \sqrt{\frac{(1+\frac{i}{i+2}) \cdot \frac{2}{(i+1)(i+2)} }{w^i}} = \frac{2}{i+2} \cdot \frac{1}{w^{\frac{i}2}}
  \]
  and
  \[
  d_i = \sqrt{\frac{w^i \cdot \frac{2}{(i+1)(i+2)}}{1+\frac{i}{i+2}}} = \frac{1}{i+1} w^{\frac{i}2}
  \]
  From~\eqref{eq:rec_c}, for $i \geq 1$, we have
  \[
  c_i = \sum_{j=1}^{i} a_{j-1}d_{j-1} = 1 + \sum_{j=2}^{i}\prod_{k=1}^{j-1}b_k  =  1 + \sum_{j=1}^{i-1}\prod_{k=1}^jb_{k}.
  \]
  Then, by combining with~\eqref{eq:prod_bi}, we have
  \begin{align*}
    c_i &= \left\{
    \begin{array}{lr}
      1 + \sum_{j=1}^{i-1} \frac{2}{(j+1)(j+2)}, & w = 4  \\
      1 + \sum_{j=1}^{i-1} p^j\cdot \frac{(p-1)(p^2-1)}{(p^{j+1}-1)(p^{j+2}-1)}, & w > 4
    \end{array}
    \right. \\
    &= \left\{
    \begin{array}{lr}
      1 + \sum_{j=1}^{i-1} \frac{2}{j+1} - \frac{2}{j+2}, & w = 4  \\
      1 + \sum_{j=1}^{i-1} (p^2-1)\left( \frac{p^j}{p^{j+1}-1} - \frac{p^{j+1}}{p^{j+2}-1}\right), & w > 4
    \end{array}
    \right. \\
    &= \left\{
    \begin{array}{lr}
      2 - \frac{2}{i+1}, & w = 4  \\
      1 + p - \frac{p^i(p^2-1)}{p^{i+1}-1}, & w > 4
    \end{array}
    \right..
  \end{align*}
  This concludes the proof.
\end{proof}

The following lemma gives a lower bound on $x_i$ for any feasible solution $X$ of $(L_{m,u})$, for $i \in [1,m]$, as well as a lower bound on \text{Obj}.

\begin{lemma}
  \label{lem:main_lemma}
  For every feasible solution $X=(x_1, x_2, \ldots)$ of $(L_{m, u})$, it holds that, for $i \in [1, m]$,
  \[
  x_i \geq a_{m-i} \cdot u + b_{m-i} \cdot T_{i-1} \quad \textrm{ and } \quad \text{Obj} \geq c_{m-i}\cdot u + d_{m-i}\cdot T_{i}.
  \]
  In addition, $x_i = a_{m-i} \cdot u + b_{m-i} \cdot T_{i-1}$ and $\text{Obj} = c_{m-i}\cdot u + d_{m-i}\cdot T_{i}$ if constraints $(C_j)$ for $j \in [i+1, m]$ are tight.
\end{lemma}

\begin{proof}
  The proof is by induction on $i$, for $i \in [1,m]$. The base case, namely for $i = m$, can be readily verified.
  For the inductive step, suppose that for $i \in [1, m-1]$ it holds that $x_j \geq a_{m-j}\cdot u + b_{m-j} \cdot T_{j-1}$ and $\text{Obj} \geq c_{m-j}\cdot u + d_{m-j} \cdot T_{j}$ with $j \in [i+1, m]$.
  We will show that $x_i \geq a_{m-i}\cdot u + b_{m-i}\cdot T_{i-1}$ and $\text{Obj} \geq c_{m-i}\cdot u + d_{m-i} \cdot T_{i}$. By $(C_{i+1})$, we have
  \begin{align*}
    w \cdot x_{i} &\geq T_{i+1} \\
    &= x_{i+1} + T_i \\
    &\geq a_{m-i-1} \cdot u + b_{m-i-1} \cdot T_i + T_i\\
    &= a_{m-i-1} \cdot u + (1 + b_{m-i-1}) \cdot T_i \\
    &= a_{m-i-1} \cdot u + (1 + b_{m-i-1}) \cdot (x_i + T_{i-1})
  \end{align*}
  It implies that
  \[
  x_i \geq \frac{a_{m-i-1}}{w - 1 - b_{m-i-i}} \cdot u + \frac{1 + b_{m-i-1}}{w - 1 - b_{m-i-i}} \cdot T_{i-1},
  \]
  which is equivalent to
  \[
  x_i \geq a_{m-i} \cdot u + b_{m-i} \cdot T_{i-1}.
  \]
  It is straightforward to see that the previous inequality holds with equality if constraints $(C_j)$ are tight for $j \in [i+1, m]$. Moreover, from induction hypothesis, we have
  \begin{align*}
    \text{Obj} &\geq c_{m-i-1}\cdot u + d_{m-i-1}\cdot T_{i+1} \\
    &= c_{m-i-1}\cdot u + d_{m-i-1}\cdot (x_{i+1} + T_i)\\
    &\geq c_{m-i-1}\cdot u + d_{m-i-1}\cdot (a_{m-i-1}\cdot u + b_{m-i-1}\cdot T_i + T_i)\\
    &= (c_{m-i-1} + d_{m-i-1} \cdot a_{m-i-1}) \cdot u + d_{m-i-1} \cdot (1 + b_{m-i-1}) T_i,
  \end{align*}
  which is equivalent to
  \[
  \text{Obj} \geq c_{m-i}\cdot u +  d_{m-i}\cdot T_i.
  \]
  The inequality holds with equality if constraints $(C_j)$ are tight for $j \in [i+1, m]$. This concludes the proof.
\end{proof}

\begin{corollary}
  \label{cor:L_m_u_feasibility}
  If $(L_{m, u})$ is feasible, then $a_{m-1}\cdot u \leq w$.
\end{corollary}

\begin{proof}
  From Lemma~\ref{lem:main_lemma}, for any feasible solution $X(x_1, x_2, \ldots)$ of $(L_{m, u})$, it holds that $x_1 \geq a_{m-1}\cdot u$ and $x_1 \leq w$. Hence $a_{m-1}\cdot u \leq w$.
\end{proof}

Define a sequence $x^*_i$ as follows:
\[
x^*_i = w\cdot x^*_{i-1} - \sum_{j=1}^{i-1} x^*_j, \quad \textrm{ with } \quad x^*_1 = a_{m-1}\cdot u.
\]
\begin{lemma}
  $x^*_i$ has a closed formula as follows.
  \label{lem:x_star_closed_formula}
  \begin{itemize}
    \item If $w>4$, then $x^*_{i}=\alpha \cdot \rho_1^{i-1} +\beta \cdot \rho_2^{i-1}$, where $\alpha=\frac{a_{m-1}\rho_2^{m-1}-1}{\rho_2^{m-1}-\rho_1^{m-1}} \cdot u$,
    and $\beta=\frac{a_{m-1}\rho_1^{m-1}-1}{\rho_1^{m-1}-\rho_2^{m-1}} \cdot u > 0$,
    \item If $w=4$, then $x^*_{i}=(\alpha+\beta \cdot i)\cdot 2^i$, where $\alpha=\frac{2^{m-1}\cdot m \cdot a_{m-1}-1}{2^m(m-1)} \cdot u$, and $\beta=\frac{1-2^{m-1}\cdot a_{m-1}}{2^m(m-1)} \cdot u > 0$,
  \end{itemize}
  with $\rho_1=\frac{w-\sqrt{w^2-4w}}{2} > 1$ and $\rho_2=\frac{w+\sqrt{w^2-4w}}{2} > 1$, two roots of the quadratic equation $x^2-wx+w=0$. In addition, $x_i^*$ is monotone increasing in $i$ and $x_i \rightarrow +\infty$ as $i \rightarrow +\infty$.
\end{lemma}

\begin{proof}
  By definition of $x^*_i$, we have the linear recurrence relation
  \[
  x^*_i = w(x^*_{i-1} - x^*_{i-2}) \textrm{ for } i \geq 3.
  \]
  Its characteristic equation is $x^2 - wx + w = 0$. We distinguish between two cases, namely for $w = 4$ and $w > 4$.
  
  If $w > 4$, then the characteristic equation has two roots $\rho_1=\frac{w-\sqrt{w^2-4w}}{2} > 1$ and $\rho_2=\frac{w+\sqrt{w^2-4w}}{2} > 1$. It implies that
  \[
  x^*_{i}=\alpha \cdot \rho_1^{i-1} +\beta \cdot \rho_2^{i-1},
  \]
  with some coefficients $\alpha, \beta$. We can determine the value of $\alpha, \beta$ by using the fact $x^*_1 = a_{m-1}u$ and $x_m = u$ (from Lemma~\ref{lem:main_lemma}). As a result, we have
  \[
  \alpha=\frac{\rho_2^{m-1}a_{m-1}-1}{\rho_2^{m-1}-\rho_1^{m-1}} \cdot u, \textrm{ and }
  \beta=\frac{a_{m-1}\rho_1^{m-1}-1}{\rho_1^{m-1}-\rho_2^{m-1}} \cdot u.
  \]
  We will argue that $\beta > 0$. Since $\rho_2 > \rho_1$, it remains to show that $a_{m-1}\rho_1^{m-1} - 1 < 0$. From~\eqref{eq:rec_a'} and~\eqref{eq:p}, we have
  \[
  a_{m-1} = \prod_{j=1}^{m-1} \frac{1}{w-1-b_{j-1}} < \prod_{j=1}^{m-1} \frac{1}{w-1-p} = \prod_{j=1}^{m-1} \frac{p}{1+p} = \left(\frac{p}{1+p} \right)^{m-1}.
  \]
  Since $\rho_1 = \frac{w-\sqrt{w^2-4w}}{2} = 1 + p$ and $p < 1$, then
  \[
  a_{m-1}\rho_1^{m-1} - 1 < \left(\frac{p}{1+p}\right)^{m-1} \cdot (1+p)^{m-1} = p^{m-1} - 1 < 0,
  \]
  which implies that $\beta > 0$. Now, we can argue that $x^*_i$ is monotone increasing in $i$. We have
  \[
  \alpha + \beta = \frac{\rho_2^{m-1}a_{m-1}-1}{\rho_2^{m-1}-\rho_1^{m-1}} \cdot u + \frac{a_{m-1}\rho_1^{m-1}-1}{\rho_1^{m-1}-\rho_2^{m-1}} \cdot u = \frac{a_{m-1}(\rho_2^{m-1} - \rho_1^{m-1})}{\rho_2^{m-1} - \rho_1^{m-1}} \cdot u= a_{m-1} \cdot u.
  \]
  By Lemma~\ref{lem:a_b_c_d}, we have $\alpha + \beta = a_{m-1} \cdot u > 0$, which implies that, for $i \geq 1$,
  \[
  x^*_i = \alpha \cdot \rho_1^{i-1} + \beta \cdot \rho_2^{i-1} > (\alpha + \beta) \cdot \rho_1^{i-1} > 0.
  \]
  Combining with $x^*_i = w(x^*_{i-1} - x^*_{i-2})$, for $i \geq 3$, then $x^*_i$ is monotone increasing in $i$. Moreover, since $\beta > 0$ and $1< \rho_1 < \rho_2$, then $x^*_i \rightarrow +\infty$ as $i \rightarrow +\infty$. This concludes the proof of the case $w > 4$.
  
  If $w = 4$, the proof is similar to the previous one. The characteristic equation has one double root $\rho = 2$. This implies 
  \[
  x^*_{i}=(\alpha+\beta \cdot i)2^i,
  \]
  with some coefficients $\alpha, \beta$. We can determine the value of $\alpha, \beta$ by using the fact $x^*_1 = a_{m-1}u$ and $x_m = u$ (from Lemma~\ref{lem:main_lemma}). As a result, we have
  \[
  \alpha=\frac{2^{m-1}\cdot m \cdot a_{m-1}-1}{2^m(m-1)} \cdot u \quad \textrm{ and } \quad \beta=\frac{1-2^{m-1}\cdot a_{m-1}}{2^m(m-1)} \cdot u.
  \]
  We will argue that $\beta > 0$. It suffices to show that $2^{m-1}a_{m-1} - 1 < 0$. By Lemma~\ref{lem:a_b_c_d}, we have
  \[
  a_{m-1} = \frac{2}{m+1} \cdot \frac{1}{2^{m-1}} < \frac{1}{2^{m-1}},
  \]
  which implies that $\beta > 0$. Moreover, since $\alpha + \beta = \frac{2^{m-1}(m-1)\cdot a_{m-1}}{2^m(m-1)} \cdot u = \frac{a_{m-1}\cdot u}{2} > 0$, then
  \[
  x^*_{i}=(\alpha+\beta \cdot i)2^i > \alpha+\beta > 0.
  \]
  Combining with $x^*_i = w(x^*_{i-1} - x^*_{i-2})$, for $i \geq 3$, then $x^*_i$ is monotone increasing in $i$. Moreover, since $\beta > 0$, then $x^*_i \rightarrow +\infty$ as $i \rightarrow +\infty$. This concludes the proof of the case $w = 4$.
\end{proof}

We are now able to prove the main lemmas of Section~\ref{subsec:bidding.target}.

\begin{proof}[Proof of Lemma~\ref{lem:bidding.recurrence}]
  First, if $X^*$ is an optimal solution of $(L_{m, u})$, then by Corollary~\ref{cor:L_m_u_feasibility}, we have $x^*_1 = a_{m-1}\cdot u \leq w$.
  
  In addition, we will show that $X^*$ is an optimal solution of $(L_{m, u})$, if $x^*_1 = a_{m-1}\cdot u \leq w$. First we argue that $X^*$ is a feasible solution of $(L_{m, u})$. By definition of $X^*$, it satisfies constraints $(C_j)$ and $x_m = u$. Lemma~\ref{lem:x_star_closed_formula} shows that $x^*_i$ is monotone increasing in $i$. Moreover, from Lemma~\ref{lem:main_lemma}, $X^*$ is optimal.
\end{proof}

\begin{proof}[Proof of Lemma~\ref{lem:bidding.r}]
  By Lemma~\ref{lem:main_lemma} and Theorem~\ref{lem:bidding.recurrence}, we have $r^*_{m,u} = c_m$. Combining with Lemma~\ref{lem:a_b_c_d}, we have
  \[
  r^*_{m,u} = \left\{
  \begin{array}{lr}
    2 - \frac{2}{m+1}, & w = 4  \\
    1 + p - \frac{p^m(p^2-1)}{p^{m+1}-1}, & w > 4
  \end{array}
  \right.,
  \]
  which implies that, the worst case ratio is
  \begin{align*}
    r^* &= \lim_{m \rightarrow +\infty} r^*_{m,u}
    = \lim_{m \rightarrow +\infty} 1 + p - \frac{p^m(p^2-1)}{p^{m+1}-1}
    = 1 + p  \\
    &= \frac{w - \sqrt{w^2 - 4w}}{2}.
  \end{align*}
\end{proof}

\end{document}